\documentclass[11pt,fleqn]{article}

\usepackage{amsmath,amssymb,amsthm}
\usepackage{graphicx}
\usepackage{color}
\usepackage{multirow}

\usepackage{hyperref}
\hypersetup{colorlinks, linkcolor=blue, citecolor=blue, urlcolor=blue}

\allowdisplaybreaks

\newcommand{\p}{\partial}

\newcommand{\rank}{\mathop{\rm rank}\nolimits}
\newcommand{\pr}{{\rm P}}
\newcommand{\checked}[1][\null]{\ensuremath{\boldsymbol{\surd}}}

\newcommand{\todo}[1][\null]{\ensuremath{\clubsuit}}

\newcommand{\noprint}[1]{}

\newtheorem{theorem}{Theorem}

\newtheorem{corollary}{Corollary}
\newtheorem{proposition}{Proposition}
{\theoremstyle{definition}
\newtheorem{definition}{Definition}

\newtheorem*{remark*}{Remark}
}

\newcommand{\JJ}{\mathcal{J}}
\newcommand{\vv}{\mathbf{v}}

\newcommand{\nn}{\mathbf{\nabla}}

\newcommand{\ve}{\varepsilon}
\newcommand{\Ad}{\mathrm{Ad}}

\newcommand{\ZZ}{\mathcal{Z}}
\newcommand{\YY}{\mathcal{Y}}
\newcommand{\XX}{\mathcal{X}}
\newcommand{\DDD}{\mathcal{D}}

\begin{document}

\par\noindent {\LARGE\bf
Symmetry preserving parameterization schemes
\par}
{\vspace{4mm}\par\noindent {\bf Roman O.\ Popovych~$^\dag$ and Alexander Bihlo~$^\ddag$
} \par\vspace{2mm}\par}

{\vspace{2mm}\par\noindent {\it
$^\dag$~Institute of Mathematics of NAS of Ukraine, 3 Tereshchenkivska Str., 01601 Kyiv, Ukraine\\
$\phantom{^\dag}$~~Wolfgang Pauli Institute, Nordbergstra{\ss}e 15, A-1090 Vienna, Austria\\
}}
{\noindent \vspace{2mm}{\it
$\phantom{^\dag}$~\textup{E-mail}: rop@imath.kiev.ua
}\par}

{\vspace{2mm}\par\noindent {\it
$^{\ddag}$~Centre de recherches math\'{e}matiques, Universit\'{e} de Montr\'{e}al, C.P.\ 6128, succ.\ Centre-ville,\\
$\phantom{^\ddag}$~Montr\'{e}al (QC) H3C 3J7, Canada\\
}}
{\noindent \vspace{2mm}{\it
$\phantom{^\dag}$~\textup{E-mail}: alexander.bihlo@univie.ac.at
}\par}

\vspace{5mm}\par\noindent\hspace*{8mm}\parbox{140mm}{\small
Methods for the design of physical parameterization schemes that possess certain invariance properties are discussed. These methods are based on different techniques of group classification and provide means to determine expressions for unclosed terms arising in the course of averaging of nonlinear differential equations. The demand that the averaged equation is invariant with respect to a subalgebra of the maximal Lie invariance algebra of the unaveraged equation leads to a problem of inverse group classification which is solved by the description of differential invariants of the selected subalgebra. Given no prescribed symmetry group, the direct group classification problem is relevant. Within this framework, the algebraic method or direct integration of determining equations for Lie symmetries can be applied. For cumbersome parameterizations, a preliminary group classification can be carried out. The methods presented are exemplified by parameterizing the eddy vorticity flux in the averaged vorticity equation. In particular, differential invariants of (infinite dimensional) subalgebras of the maximal Lie invariance algebra of the unaveraged vorticity equation are computed. A hierarchy of normalized subclasses of generalized vorticity equations is constructed. Invariant parameterizations possessing minimal symmetry extensions are described and a restricted class of invariant parameterization is exhaustively classified. The physical importance of the parameterizations designed is discussed.
}\par\vspace{5mm}


\section{Introduction}

The problem of parameterization is one of the most important issues in modern dynamic meteorology and climate research~\cite{kaln02Ay,sten07Ay}. As even the most accurate present days numerical models are not capable to resolve all small scale features of the atmosphere, there is a necessity for finding ways to incorporate these unresolved processes in terms of the resolved ones. This technique is referred to as parameterization. The physical processes being parameterized in numerical weather and climate prediction models can be quite different, including e.g.\ cumulus convection, momentum, heat and moisture fluxes, gravity wave drag and vegetation effects. The general problem of parameterization is intimately linked to the design of closure schemes for averaged (or filtered) nonlinear equations. By averaging, a nonlinear differential equation becomes unclosed, that is, there arise additional terms for which no prognostic or diagnostic equation exist. These terms must hence be re-expressed in a physically reasonable way to be included in the averaged equations.

It has been noted in~\cite{stul88Ay} that every parameterization scheme ought to retain some basic properties of the unresolved terms, which must be expressed by the resolved quantities. These properties include, just to mention a few, correct dimensionality, tensorial properties, invariance under changes of the coordinate system and invariance with respect to Galilean transformations. While the formulation of a parameterization scheme with correct dimensions is in general a straightforward task, not all parameterization schemes that have been used in practice are indeed Galilean invariant. An example for this finding is given by the classical Kuo convection scheme~\cite{kuo65Ay,kuo74Ay}. In this scheme, it is assumed that the vertically integrated time-change of the water vapor at a point locally balances a fraction of the observed precipitation rate~\cite[pp.\ 528]{eman94Ay}. This also implies that the moisture convergence is proportional to the precipitation rate. However, while the precipitation rate is clearly a Galilean invariant quantity, the moisture convergence depends on the motion of the observer~\cite{kerk06Ay}. That is, the Kuo scheme does not properly account for pure symmetry constraints, which is a potential source of unphysical effects in the results of a numerical model integration.

The latter finding is the main motivation for the present investigations. Galilean invariance is an important example for a Lie symmetry, but it is by no means the only invariance characteristic that might be of importance in the course of the parameterization process. This is why it is reasonable to focus on parameterization schemes that also preserve other symmetries. This is not an academic task. Almost all real-world processes exhibit miscellaneous symmetry characteristics. These characteristics are reflected in the symmetry properties of differential equations and correspondingly should also be reflected in case where these processes cannot be explicitly modeled by differential equations, i.e.\ in the course of parameterizations. What is hence desirable is a constructive method for the design of symmetry-preserving parameterization schemes. It is the aim of this paper to demonstrate that techniques from group analysis do provide such constructive methods. In particular, we state the following proposition:

\begin{center}
 \textit{Any problem of finding invariant parameterizations is a group classification problem.}
\end{center}

\noindent Implications following from the above proposition form the core of the present study. It appears that this issue was first opened in~\cite{ober97Ay}, dealing with the problem of turbulence closure of the averaged Navier--Stokes equations. We aim to build on this approach and extend it in several directions. As the equations of hydrodynamics and geophysical fluid dynamics usually possess wide symmetry groups \cite{andr98Ay,bihl09Ay,bila06Ay,fush94Ay,ibra95Ay}, the design of symmetry-preserving parameterizations will in general lead to a great variety of different classes of invariant schemes.

\looseness=1
Needless to say that the parameterization problem is too comprehensive both in theory and applications to be treated exhaustively in a single paper. Therefore, it is crucial to restrict to a setting that allows to demonstrate the basic ideas of invariant parameterizations without overly complicating the presentation by physical or technical details. This is the reason for illustrating the invariant parameterization procedure with the rather elementary barotropic vorticity equation. For the sake of simplicity, we moreover solely focus on local closure schemes in the present study. That is, the quantities to be parameterized at each point are substituted with known quantities defined at the same respective point~\cite{stul88Ay}. This renders it possible to thoroughly use differential equations and hence it will not be necessary to pass to integro-differential equations, as would be the case for nonlocal closure schemes. On the other hand, this restriction at once excludes a number of processes with essential nonlocal nature, such as e.g.\ atmospheric convection. Nevertheless, there are several processes that can be adequately described within the framework of the present paper, most notably different kinds of turbulent transport \mbox{phenomena}.

The organization of this paper is the following: Section~\ref{sec:idea} discusses different possibilities for the usage of symmetry methods in the parameterization procedure, most noteworthy the application of techniques of direct and inverse group classifications. We restate some basic results from the theory of group classification and relate them to the parameterization problem. Section~\ref{sec:vorticity} is devoted to the construction of several parameterization schemes for the eddy vorticity flux of the vorticity equation using the methods introduced in the previous section. Generating sets of differential invariants and operators of invariant differentiation for subalgebras of the maximal Lie invariance algebra of the vorticity equation are computed and used in the framework of invariant parameterization (Section~\ref{sec:ParameterizationInverseGroupClassification}). It should emphasized that up to now only very few examples on exhaustive descriptions of differential invariants for infinite dimensional Lie algebras exist in the literature~\cite{cheh08Ay,golo04Ay}. A hierarchy of nested normalized subclasses of a class of generalized vorticity equations is constructed in Section~\ref{sec:AdmTrans}. Additionally, in Section~\ref{sec:EquivAlgebra} the equivalence algebras of some subclasses are directly found within the framework of the infinitesimal approach. The algebraic method of group classification is used to determine inequivalent invariant parameterization schemes. For a restricted class of generalized vorticity equations, it is proved in Section~\ref{SectionOnParameterizationViaDirectGroupClassification} that the algebraic method provides an exhaustive description of all inequivalent parameterizations of the eddy vorticity flux. For a wider class of generalized vorticity equations, in Section~\ref{sec:ParameterizationViaPreliminaryGroupClassification} we study the problem of invariant parameterization within the framework of preliminary group classification. Namely, inequivalent invariant parameterizations possessing at least one-dimensional symmetry extensions are listed. A short discussion of the results of the paper is presented in Section~\ref{sec:conclusion}, together with an outlook on forthcoming works in the field of invariant parameterization theory. In Appendix~\ref{sec:appendix1}, details on the classification of inequivalent one-dimensional subalgebras of the equivalence algebra from Theorem~\ref{TheoremOnEquivAlgebra1}, which is used in Section~\ref{sec:ParameterizationViaPreliminaryGroupClassification}, can be found.

\section{The general idea}\label{sec:idea}

Throughout the paper, the notation we adopt follows closely that presented in the textbook~\cite{olve86Ay}. Let there be given a system of differential equations
\begin{equation}\label{eq:generalde}
    \Delta^l(x,u_{(n)}) = 0, \qquad l=1,\dots,m,
\end{equation}
where $x=(x^1,\dots,x^p)$ denote the independent variables and the tuple $u_{(n)}$ includes all dependent variables $u=(u^1,\dots,u^q)$ as well as all derivatives of $u$ with respect to $x$ up to order~$n$.
Hereafter, subscripts of functions denote differentiation with respect to the corresponding variables.

Both numerical representations of~\eqref{eq:generalde} as well as real-time measurements are not able to capture the instantaneous value of $u$, but rather only provide some mean values. That is, to employ~\eqref{eq:generalde} in practice usually requires an averaging or filter procedure. For this purpose, $u$ is separated according to
\[
    u = \bar u + u',
\]
where $\bar u$ and $u'$ refer to the averaged and the deviation quantities, respectively. The precise form of the averaging or filter method used determines additional calculation rules, e.g., $\overline{ab}=\bar a \bar b + \overline{a'b'}$ for the classical Reynolds averaging. At the present stage it is not essential to already commit oneself to a definite averaging method. For nonlinear system~\eqref{eq:generalde} averaging usually gives expressions
\begin{equation}\label{eq:generaldea}
    \tilde \Delta^l(x,\bar u_{(n)},w)=0, \qquad l=1,\dots,m,
\end{equation}
where $\tilde \Delta^l$ are smooth functions of their arguments whose explicit form is precisely determined by the form of $\Delta^l$ and the chosen averaging rule. The tuple $w=(w^1,\dots,w^k)$ includes all averaged nonlinear combinations of terms, which cannot be obtained by means of the quantities $\bar u_{(n)}$. These combinations typically include such expressions as $\overline{u'u'}$, $\overline{u'\bar u}$, $\overline{u'u_x'}$, etc., referred to as subgrid scale terms. Stated in another way, system~\eqref{eq:generaldea} contains more unknown quantities than equations. To solve system~\eqref{eq:generaldea}, suitable assumptions on $w$ have to be made. An adequate choice for these assumptions is the problem of parameterization.

The most straightforward way to tackle this issue is to directly express the unclosed terms $w$ as functions of the variables $x$ and $\bar u_{(r)}$ for some~$r$ which can be greater than~$n$. In other words, system~\eqref{eq:generaldea} is closed via
\begin{equation}\label{eq:generaldec}
    \tilde \Delta^l(x,\bar u_{(n)},f(x,\bar u_{(r)}))=0, \qquad l=1,\dots,m,
\end{equation}
using the relation $w^s=f^s(x,\bar u_{(r)})$, $s=1,\dots,k$, where $k$ is the number of unclosed terms which are necessary to be parameterized. The purpose of this paper is to discuss different paradigms for the choice of the functions $f=(f^1,\dots,f^k)$ within the symmetry approach. In other words, we should carry out, in different ways, group analysis of the class~\eqref{eq:generaldec} with the arbitrary elements running through a set of differential functions. To simplify notation, we will omit bars over the dependent variables in systems where parameterization of $w$ is already applied.

\begin{remark*}
In the theory of group classification, any class of differential equations is considered in a jet space of a fixed order. That is, both the explicit part of the expression of the general equation from the class and the arbitrary elements can be assumed to depend on derivatives up to the same order. In contrast to this, for the construction of parameterization schemes it is beneficial to allow for varying the orders of arbitrary elements while the order of the explicitly resolved terms is fixed. This is why we preserve different notations for the orders of derivatives in the explicit part of the expression of the general equation and in the arbitrary elements of the class~\eqref{eq:generaldec}.
\end{remark*}

\subsection{Parameterization via inverse group classification}\label{sec:ParameterizationViaInverseGroupClassification}

Parameterizations based on Lie symmetries appear to have been first investigated for the Navier--Stokes equations. It was gradually realized that the consideration of symmetries plays a key role in the construction of subgrid scale models for the Navier--Stokes equations to allow for realistic simulations of flow evolution. See~\cite{ober97Ay,ober00Ay} for a further discussions on this subject. The approach involving symmetries for the design of local closure schemes, was later extended in~\cite{raza06Ay,raza07Ay,raza07By} in order to incorporate also the second law of thermodynamics into the consideration.

For an arbitrary system of differential equations, this approach can be sketched as follows: First, determine the Lie symmetry group~$G$ (resp.\ the corresponding Lie invariance algebra~$\mathfrak g$) of the model to be investigated. For common models of hydro-thermodynamics these computations were already carried out and results can be found in collections like~\cite{ibra95Ay}. Subsequently, determine the differential invariants of the group~$G$. If the left hand side of system~\eqref{eq:generaldea} is formulated in terms of these invariants by an adequate choice of the function $f$, it is guaranteed that the parameterized system will admit the same group of point symmetries as the unfiltered system. Usually this leads to classes of differential equations rather than to a single model. That is, among all models constructed this way it is be possible to select those which also satisfy other desired physical and mathematical properties.

The procedure outlined above can be viewed as a special application of techniques of \textit{inverse group classification}. Inverse group classification starts with a prescribed symmetry group and aims to determine the entire class of differential equations admitting the given group as a symmetry group \cite{ovsi82Ay}. Thus, in \cite{ober97Ay,ober00Ay,raza06Ay,raza07Ay,raza07By} it is assumed that the closure scheme for the subgrid scale terms leads to classes of differential equations admitting the complete Lie symmetry group of the Navier--Stokes equations. From the mathematical point of view, this assumption is justified as filtering (or averaging) of the Navier--Stokes equations introduces a turbulent friction term among the viscous friction term that already appears in the unfiltered equations. That is, filtering does not principally perturb the structure of the Navier--Stokes equations. However, this assumption may not be as well justified if a model is chosen, where filtering leads to terms of forms not already included in the unfiltered model. In such cases, it may be more straightforward to solve the parameterization problem by inverse group classification only with respect to particular subgroups of the Lie symmetry group~$G$ of the initial system~$\mathcal S$ of differential equations.
The selection of proper subgroups of~$G$ can be realized involving physical arguments.

Another possible way for such a selection may be related to boundary-value problems.
One can choose a subgroup of~$G$ consisting of either symmetries of a particular boundary-value problem for~$\mathcal S$ or equivalence transformations of a relevant class~$\mathcal B$ of similar boundary-value problems for~$\mathcal S$.
The re-interpretation of symmetries of~$\mathcal S$ as equivalence transformations for~$\mathcal B$ is natural because they often have a clear physical significance, such as the rescaling of a domain (e.g.\ when conducting numerical tests), shifts of space and time variables or the transformation from a resting reference frame to reference frames moving with constant velocity, and all of these fundamental symmetries are usually broken when considering a fixed boundary-value problem.
A scaling symmetry of~$\mathcal S$ is restored as an equivalence transformation for~$\mathcal B$ if the class~$\mathcal B$ consists of boundary-value problems of all possible domain sizes. The same argument holds for shifts and  Galilean boosts. This re-interpretation does not change the general algorithm for the construction of parameterization schemes using inverse group classification. It rather requires an analysis of the parameterization problem to be treated using group classification methods. The argument is to determine which symmetries map any particular boundary-value problem from~$\mathcal B$ to another problem from~$\mathcal B$. The symmetries fulfilling this requirement are to be interpreted as equivalence transformations for the given class~$\mathcal B$ of boundary-value problems. The symmetries not compatible with~$\mathcal B$ could therefore be excluded from the consideration.

The approach of inverse group classification usually relies on the notion of differential invariants \cite{olve95Ay,ovsi82Ay}. Differential invariants are defined as the invariants of the prolonged action of a given symmetry group. They can be determined either with the infinitesimal method~\cite{golo04Ay,ovsi82Ay} or with the technique of moving frames~\cite{cheh08Ay,fels98Ay,olve07Ay}. In the present paper we will use the former method which is briefly described here for this reason.

Let $X$ be the $p$-dimensional space of independent and $U$ be the $q$-dimensional space of dependent variables. The connected Lie group $G$ acts locally as a point transformation group on the space $J^0=X\times U$, with $\mathfrak g$ denoting the associated Lie algebra of infinitesimal generators. (The whole consideration is assumed local). Each element of $\mathfrak g$ is of the form $Q=\xi^i(x,u)\p_{x^i}+\varphi^a(x,u)\p_{u^a}$. In this section the indices $i$ and $j$ run from $1$ to $p$ while the indices $a$ and $b$ run from $1$ to $q$, and the summation convention over repeated indices is used. The space $J^r=X\times U_{(r)}$ is the $r$th prolongation of the space $X\times U$ (the $r$th order jet space), which is the space endowed with coordinates $x_i$ and $u^a_\alpha$, $|\alpha|:=\alpha_1+\cdots+\alpha_p<r$, where $u^a_\alpha$ stands for the variable corresponding to the derivative $\p^{|\alpha|}u^a/\p x_1^{\alpha_1}\ldots\p x_p^{\alpha_p}$, and $\alpha=(\alpha_1,\ldots,\alpha_p)$ is an arbitrary multiindex, $\alpha_i\in\mathbb N_0=\mathbb N\cup\{0\}$. The action of $G$ can be extended to an action on $J^r$ and so the elements of $\mathfrak g$ can be prolonged via
\begin{equation}\label{EqProlongedOp}
Q_{(r)} = Q+\sum_{\alpha>0}\varphi^{a\alpha}\p_{u^a_\alpha}, \quad
    \varphi^{a\alpha} := \mathrm D_1^{\alpha_1}\cdots\mathrm D_p^{\alpha_p}(\varphi^a-\xi^i u_i^\alpha) + \xi^i u^a_{\alpha+\delta_i}.
\end{equation}
Here $\mathrm D_i=\mathrm D_{x_i}$ denotes the operator of total differentiation with respect to the variable~$x_i$, i.e.
$\mathrm D_i=\p_{x_i}+u^a_{\alpha+\delta_i}\p_{u^a_\alpha}$, where
$\delta_i$ is the multiindex whose $i$th entry equals 1 and whose other entries are zero.
More details can be found in the textbooks~\cite{olve95Ay,olve86Ay,ovsi82Ay}.

A differential function $f$ (i.e., a smooth function from $J^r$ to $\mathbb R$ for some~$r$) is called an ($r$th order) \emph{differential invariant} of the group $G$ if for any transformation $g\colon (x,u)\mapsto (\tilde x,\tilde u)$ from $G$ we have that $f(\tilde x, \tilde u_{(r)}) = f(x,u_{(r)})$.
The function~$f$ is a differential invariant of~$G$ if and only if the equality $Q_{(r)}f = 0$ holds for any  $Q\in\mathfrak g$.
A vector field $\mathfrak d$ defined in the infinite jet space~$J^\infty$, which is the inverse limit of the sequence of natural projections from $J^{p+1}$ to~$J^p$ for $p\in\mathbb N_0$,
is called an \emph{operator of invariant differentiation} for the group $G$ if the result $\mathfrak df$ of its action to any differential invariant $f$ of~$G$ also is a differential invariant of~$G$.

The \emph{Fundamental Basis Theorem} states that any finite-dimensional Lie group (or, more generally, any Lie pseudo-group satisfying certain condition) acting on~$J^0$ possesses exactly $p$ operators of invariant differentiation, which are independent up to linear combining with coefficients depending on differential invariants, and a finite basis of differential invariants, i.e., a finite set of differential invariants such that any differential invariant of the group can be obtained from basis invariants by a finite number of functional operations and actions by the chosen independent operators of invariant differentiation.

For a vector field $\mathfrak d$ in~$J^\infty$ to be an operator of invariant differentiation of~$G$, it is sufficient that it commutes with every infinitely prolonged operator from the corresponding Lie algebra~$\mathfrak g$, i.e., $[\mathfrak d,Q_{(\infty)}]=0$ for any $Q\in\mathfrak g$. If the group~$G$ is finite dimensional, a set of $p$ independent operators of invariant differentiation can be found in the form $\mathfrak d=h^i\mathrm D_i$ by solving, with respect to the differential functions $h^i=h^i(x,u_{(r)})$, the system of first-order quasi-linear partial differential equations
\[
Q_{(r)}h^i = h^j\mathrm D_j\xi^i,
\]
where $Q$ runs through a basis of the corresponding Lie algebra~$\mathfrak g$ and $r$ equals the minimum order for which the rank of the prolonged basis operators of~$\mathfrak g$ coincides with its dimension.
Eventually, it may be convenient to determine $h^i$ in the implicit form $\Omega^j(h^1,\dots,h^p,x,u_{(r)})=0$,
where $\det(\Omega^j_{h_i})\ne0$ and $\Omega^j$ satisfy the associated system of homogeneous equations
\[
  \left(Q_{(r)}+(h^{i'}\mathrm D_{i'}\xi^i)\p_{\lambda^i}\right)\Omega^j = 0.
\]
In the infinite dimensional case, the construction of invariant differentiation operators is analogous though more sophisticated.

A systematic approach to parameterization via inverse group classification hence consists of determining the basis differential invariants of a group together with the list of operators of invariant differentiation. Subsequently, there are infinitely many parameterizations that can be constructed, which admit the given group as a symmetry group.

\subsection{Parameterization via direct group classification}

The main assumption in the approach presented in~\cite{ober97Ay,raza07By} is that a realistic subgrid scale model for the Navier--Stokes equations should admit the symmetry group of the original equations. However, this assumption is rather restrictive in more general situations. While it is true that a filtered model should be a realistic approximation of the unfiltered equations, parameterization schemes also have to take into account physical processes for which we may not have a precise understanding yet. That is, one eventually has to face the problem to deal with processes for which we may not even have a differential equation. This particularly means that a fixed set of symmetries (as for the Navier--Stokes equations) may not be obtainable.

On the other hand, symmetries do provide a useful guiding principle for the selection of physical models. As nature tends to prefer states with a high degree of symmetry, a general procedure for the derivation of symmetry-preserving parameterization schemes seems reasonable. The only crucial remark is, that we may not know in advance, which symmetries are most essential for capturing the characteristics of the underlying physical processes. For such problems, application of inverse group classification techniques is at once limited. Rather, it may be beneficial to derive parameterization schemes admitting different symmetry groups and subsequently test these various schemes to select among them those which best describe the processes under consideration. That is, instead of expressing the tuple~$w$ in system~\eqref{eq:generaldea} using differential invariants of a symmetry group of the unfiltered equations (or another convenient symmetry group) from the beginning, we investigate symmetries of system~\eqref{eq:generaldec} for different realizations of the functions $f$ which are eventually required to satisfy some prescribed conditions. This way, we could be interested in special classes of parameterizations, such as time- or spatially independent ones. This naturally leads back to the usual problem of \textit{direct group classification}: Let there be given a class of differential equations, parameterized by arbitrary functions. First determine the symmetries admitted for all choices of these functions, leading to the kernel of symmetry groups of the class under consideration. Subsequently, investigate for which special values of these parameter-functions there are extensions of the kernel group~\cite{ovsi82Ay,popo10Ay}.

To systematically carry out direct group classification, it is necessary to determine the equivalence group of the class, i.e.\ the group of transformations mapping an equation from the class~\eqref{eq:generaldec} to an equation from the same class. Classification of extensions of the kernel group is then done up to equivalence imposed by the equivalence group of the class~\eqref{eq:generaldec}.

The continuous part of the equivalence group can be found using infinitesimal methods in much the same way as Lie symmetries can be found using the infinitesimal invariance criterion. This firstly yields the equivalence algebra, the elements of which can then be integrated to give the continuous equivalence group. See \cite{ovsi82Ay,popo10Ay} for more details on this subject.

We now formalize the method reviewed in the previous paragraphs. Let there be given a class of differential equations of the form~\eqref{eq:generaldec}, $\tilde\Delta^l(x,\bar u_{(n)},f(x,\bar u_{(r)}))=0$, $l=1,\dots,m$. The arbitrary elements $f$ usually satisfy an auxiliary system of equations $S(x,u_{(r)},f_{(\rho)}(x,u_{(r)}))=0$, $S=(S^1,\dots,S^s),$ and an inequality $\Sigma(x,u_{(r)},f_{(\rho)}(x,u_{(r)}))\ne0$, where $f_{(\rho)}$ denotes the collection of $f$ and all derivatives of $f$ with respect to the variables $x$ and $u_{(r)}$ up to order $\rho$. The conditions $S=0$ and $\Sigma\ne0$ restrict the generality of $f$ and hence allow the design of specialized parameterizations.
We denote the solution set of the auxiliary system by~$\mathcal S$, the system of form~\eqref{eq:generaldec} corresponding to an $f\in\mathcal S$ by~$\mathcal L_f$
and the entire class of such system by~$\mathcal L|_{\mathcal S}$.

The set of all (nondegenerate) point transformations that map a system $\mathcal L_f$ to a system $\mathcal L_{\tilde f}$, where both $\smash{f,\tilde f\in\mathcal S}$ is denoted by $\smash{\mathrm T(f,\tilde f)}$ and is referred to as the set of admissible transformations from the system $\mathcal L_f$ to the system $\mathcal L_{\tilde f}$. The collection of all point transformations relating at least two systems from the  class $\mathcal L|_S$ gives rise to the set of admissible transformations of~$\mathcal L|_S$.

\begin{definition}
The set of admissible transformations of the class $\mathcal L|_S$ is the set  $\mathrm T(\mathcal L|_S)=\{(f,\tilde f,\varphi)\, |\, f,\tilde f\in \mathcal S, \varphi \in \mathrm T(f,\tilde f)\}$.
\end{definition}

That is, an admissible transformation is a triple, consisting of the initial system (with arbitrary elements $f$), the target system (with arbitrary elements $\tilde f$) and a mapping $\varphi$ between these two systems. 
It is obvious that the usual composition of mappings defines the groupoid structure on the set $\mathrm T(\mathcal L|_S)$ and  
the general point equivalence between equations from the class $\mathcal L|_S$ coincides with that generated by elements of $\mathrm T(\mathcal L|_S)$.
This is why we can also call $\mathrm T(\mathcal L|_S)$ the \emph{equivalence groupoid} of the class $\mathcal L|_S$.

The \emph{usual equivalence group} $G^{\sim}=G^{\sim}(\mathcal L|_{\mathcal S})$ of the class~$\mathcal L|_{\mathcal S}$ is defined in a rigorous way
in terms of admissible transformations.
Namely, any element~$\Phi$ from~$G^{\sim}$ is a point transformation in the space of $(x,u_{(r)},f)$,
which is projectable on the space of $(x,u_{(r')})$ for any $0\le r'\le r$,
so that the projection is the $r'$th order prolongation of $\Phi|_{(x,u)}$, the projection of $\Phi$ on the variables $(x,u)$,
and for any arbitrary elements $f\in\mathcal S$ we have that $\Phi f\in\mathcal S$ and $\Phi|_{(x,u)}\in\mathrm T(f,\Phi f)$.
The admissible transformations of the form $(f,\Phi f,\Phi|_{(x,u)})$,
where $f\in\mathcal S$ and $\Phi\in G^\sim$, are called induced by transformations from the equivalence group $G^{\sim}$.
Needless to say, that in general not all admissible transformations are induced by elements from the equivalence group.
Different generalizations of the notion of usual equivalence groups exist in the literature
\cite{mele96Ay,popo10Ay}. By~$\mathfrak g^\sim$ we denote the algebra associated with the equivalence group~$G^\sim$ and call it the equivalence algebra of the class~$\mathcal L|_{\mathcal S}$.

After clarifying the notion of admissible transformations and equivalence groups, we move on with the description of a common technique in the course of group analysis of differential equations, namely the algebraic method. Within this method one at first should classify inequivalent subalgebras of the corresponding equivalence algebra and then solve the inverse group classification problem for each of the subalgebras obtained. This procedure usually yields most of the cases of extensions and therefore leads to preliminary group classification (see, e.g., \cite{ibra92Ay,ibra91Ay,torr96Ay} for applications of this technique to various classes of differential equations).

The algebraic method rests on the following two propositions~\cite{card11Ay}:

\begin{proposition}
Let $\mathfrak a$ be a subalgebra of the equivalence algebra~$\mathfrak g^\sim$
of the class $\mathcal L|_{\mathcal S}$, $\mathfrak a\subset\mathfrak g^\sim$,
and let $f^0(x,u_{(r)})\in\mathcal S$ be a value of the tuple of arbitrary elements~$f$
for which the algebraic equation $f=f^0(x,u_{(r)})$ is invariant with respect to~$\mathfrak a$.
Then the differential equation $\mathcal L|_{f^0}$ is invariant with respect to the projection of $\mathfrak a$ to the space of variables $(x,u)$.
\end{proposition}

\begin{proposition}
Let $\mathcal S_i$ be the subset of~$\mathcal S$ that consists of all arbitrary elements for which the corresponding algebraic equations are invariant with respect to the same subalgebra of the equivalence algebra~$\mathfrak g^\sim$ and
let $\mathfrak a_i$ be the maximal subalgebra of~$\mathfrak g^\sim$ for which $\mathcal S_i$ satisfies this property, $i=1,2$.
Then the subalgebras~$\mathfrak a_1$ and~$\mathfrak a_2$ are equivalent with respect to the adjoint action of~$G^\sim$ if and only if the subsets~$\mathcal S_1$ and~$\mathcal S_2$ are mapped to each other by transformations from~$G^\sim$.
\end{proposition}

The result of preliminary group classification is a list of inequivalent (with respect to the equivalence group) members $\mathcal L_f$ of the class $\mathcal L|_{\mathcal S}$, admitting symmetry extension of the kernel of symmetry algebras using subalgebras of the equivalence algebra.

Although the algebraic method is a straightforward tool to derive cases of symmetry extensions for classes of differential equations with arbitrary elements, there remains the important question when it gives complete group classification, i.e., preliminary and complete group classification coincide. This question is of importance also for the problem of parameterization, as only complete group classification will lead to an exhaustive description of all possible parameterization schemes feasible for some class of differential equations.
The answer is that the class under consideration should be \emph{weakly normalized in infinitesimal sense}, i.e., it should satisfy the following property: The span of maximal Lie invariance algebras of all equations from the class is contained in the projection of the corresponding equivalence algebra to the space of independent and dependent variables,
\[\langle\mathfrak g_f\mid f\in \mathcal S\rangle\subset\pr\mathfrak g^\sim.\]
At the same time, it is better to use a stronger notion of normalization introduced in~\cite{popo10Ay}.

\begin{definition}\label{DefinitionOfNormalization}
The class $\mathcal L|_{\mathcal S}$ is \emph{normalized} if 
its equivalence groupoid is generated by its equivalence group, i.e.\
$\forall\, (f,\tilde f, \varphi)\in\mathrm T(\mathcal L|_{\mathcal S})$ $\exists\Phi\in G^\sim$: $\tilde f = \Phi f$ and $\varphi=\Phi|_{(x,u)}$.
\end{definition}

The normalization of~$\mathcal L|_{\mathcal S}$ in the sense of Definition~\ref{DefinitionOfNormalization} additionally implies that
the group classification of equations from this class up to $G^\sim$-equivalence coincides with the group classification using the general point transformation equivalence. Due to this fact we have no additional equivalences between cases obtained under the classification up to~$G^\sim$-equivalence. As a result, solving the group classification problem for normalized classes of differential equations is especially convenient and effective.

In turn, depending on normalization properties of the given class (or their lacking), different strategies of group classification should be applied~\cite{popo10Ay}. For a normalized class, the group classification problem is reduced, within the infinitesimal approach, to classification of subalgebras of its equivalence algebra \cite{basa01Ay,lahn06Ay,popo10Ay,zhda99Ay}. A class that is not normalized can eventually be embedded into a normalized class which is not necessarily minimal among the normalized superclasses~\cite{popo10Ay,popo10By}. One more way to treat a non-normalized class is to partition it into a family of normalized subclasses and to subsequently classify each subclass separately~\cite{popo10Ay,vane09Ay}. If a partition into normalized subclasses is difficult to construct due to the complicated structure of
the set of admissible transformations, conditional equivalence groups and additional equivalence transformations may be involved in
the group classification~\cite{ivan10Ay,popo04Ay,vane09Ay}. In the case when the class is parameterized by constant arbitrary elements or arbitrary elements depending only on one or two arguments, one can apply the direct method of group classification based on compatibility analysis and integration of the determining equations for Lie symmetries up to~$G^\sim$-equivalence~\cite{akha91Ay,niki01Ay,ovsi82Ay}. Recall that these determining equations involve both coefficients of a Lie symmetry operator of a system~$\mathcal L_f$ and the corresponding tuple of arbitrary elements~$f$ and follow from the infinitesimal invariance criterion~\cite{olve86Ay,ovsi82Ay},
\[
 Q_{(r')}\tilde\Delta^l(x,\bar u_{(n)},f(x,\bar u_{(r)}))|_{\mathcal L_f}^{ } = 0,\quad l=1,\dots,m.
\]
Here $r'=\max\{n,r\}$, the prolongation $Q_{(r')}$ of~$Q$ is defined by~\eqref{EqProlongedOp} and the symbol $|_{\mathcal L_f}$ means that above relation holds on solutions of the system~$\mathcal L_f$.

\section{\texorpdfstring{Symmetry preserving parameterizations\\ for vorticity equation}{Symmetry preserving parameterizations for vorticity equation}}\label{sec:vorticity}

The inviscid barotropic vorticity equations in Cartesian coordinates reads
\begin{equation}\label{eq:vort}
    \zeta_t + \{\psi,\zeta\} = 0
\end{equation}
where $\{a,b\}=a_xb_y-a_yb_x$ denotes the usual Poisson bracket with respect the variables~$x$ and~$y$. The vorticity $\zeta$ and the stream function $\psi$ are related through the Laplacian, i.e.\ $\zeta = \nn^2\psi$. The two-dimensional wind field $\vv=(u,v,0)^{\mathrm T}$ is reconstructed from the stream function via the relation $\vv = \mathbf{k}\times \nn \psi$, where $\mathbf{k}$ is the vertical unit vector.

The maximal Lie invariance algebra~$\mathfrak g_0$ of the equation~\eqref{eq:vort} is generated by the operators
\begin{align}\label{eq:algebravort}
\begin{split}
   &\DDD_1 = t\p_t - \psi\p_{\psi},\quad  \p_t, \quad \DDD_2=x\p_x+y\p_y+2\psi\p_{\psi}, \\
   &\JJ = -y\p_x + x\p_y,\quad \JJ^t=-ty\p_x+tx\p_y+\tfrac12(x^2+y^2)\p_{\psi}, \\
   &\XX(\gamma^1) = \gamma^1(t)\p_x - \gamma^1_t(t)y\p_{\psi},\quad \YY(\gamma^2) = \gamma^2(t)\p_y+\gamma^2_t(t)x\p_{\psi},\\
   &\ZZ(\chi) = \chi(t)\p_{\psi},
\end{split}
\end{align}
where $\gamma^1$, $\gamma^2$ and $\chi$ run through the set of smooth functions of~$t$.
See, e.g.~\cite{andr98Ay,bihl09Ay} for further discussions.

Reynolds averaging the above equation leads to
\begin{equation}\label{eq:vorta}
    \bar \zeta_t + \{\bar \psi,\bar \zeta\} = \nn\cdot(\overline{\vv'\zeta'}).
\end{equation}
The term $\overline{\vv'\zeta'}=(\overline{u'\zeta'},\overline{v'\zeta'},0)^{\mathrm T}$ is the horizontal eddy vorticity flux. Its divergence provides a source term for the averaged vorticity equation.  The presence of this source term destroys several of the properties of~\eqref{eq:vort}, such as, e.g., possessing conservation laws. In this paper we aim to find parameterizations of this flux term, which admit certain symmetries.

A simple choice for a parameterization of the eddy vorticity flux is given by the \textit{down-gradient ansatz}
\[
    \overline{\vv'\zeta'} = -K\nn\overline{\zeta},
\]
where the eddy viscosity coefficient $K$ still needs to be specified. Physically, this ansatz accounts for the necessity of the vorticity flux to be directed down-scale, as enstrophy (integrated squared vorticity) is continuously dissipated at small scales. Moreover, this ansatz will lead to a uniform distribution of the mean vorticity field, provided there is no external forcing that counteracts this tendency~\cite{mars84Ay}. The simplest form of the parameter $K$ is apparently $K=K(x,y)$, i.e.\ the eddy viscosity coefficient is only a function of space. More advanced ansatzes for $K$ assume dependence on $\overline{\zeta'^2}$, which is the eddy enstrophy~\cite{mars84Ay} (see also the discussion in the recent paper~\cite{mars10Ay}). This way, the strength of the eddy vorticity flux depends on the intensity of two-dimensional turbulence, which gives a more realistic model for the behavior of the fluid. There also exist a number of other parameterization schemes that can be applied to the vorticity equation, such as methods based on statistical mechanics~\cite{krai76Ay} or the anticipated potential vorticity method~\cite{sado85Ay,vall88Ay}.

\looseness=1
In the present framework, we exclusively focus on first order closure schemes. This is why we are only able to parameterize the eddy vorticity flux using the independent and dependent variables, as well as all derivatives of the dependent variables. This obviously excludes the more sophisticated and recent parameterization ansatzes of geophysical fluid dynamics from the present study. On the other hand, the basic method of invariant parameterization can already be demonstrated for this rather simple model. Indeed, symmetries of the vorticity equation employing the down-gradient ansatz or related parameterizations are investigated below using both inverse and direct group classification. Physically more advanced examples for parameterizations can be constructed following the methods outlined in Section~\ref{sec:idea} and exemplified subsequently.

\subsection{Parameterization via inverse group classification}\label{sec:ParameterizationInverseGroupClassification}

This is the technique by~\cite{ober97Ay,raza07By} applied to the inviscid vorticity equation. In view of the description of Section~\ref{sec:ParameterizationViaInverseGroupClassification} this approach consists of singling out subgroups (subalgebras) of the maximal Lie invariance group (algebra) of the vorticity equation and computation of the associated differential invariants (via a basis of differential invariants and operators of invariant differentiation). These differential invariants can then be used to construct different parameterizations of the eddy vorticity flux.

It is important to note that singling out subgroups of the maximal Lie invariance group of the vorticity equation is a meteorological way of group classification. This is why it is necessary to have a basic understanding of the processes to be parameterized before the selection of a particular group is done (otherwise, we would have to face the problem of how to combine these invariants to physically meaningful parameterizations). For the vorticity equation, we demonstrate the basic mechanisms of parameterizations via inverse group classification by singling out subgroups that allow to include the above down-gradient ansatz. This choice is of course not unique as there exist various other possibilities for parameterizations of the eddy vorticity flux. However, this choice allows us to demonstrate several of the issues of parameterization via inverse group classification.

\medskip

\noindent\textbf{Invariance under the whole Lie symmetry group.}
To present differential invariants of the whole Lie invariance algebra~$\mathfrak g_0$,
we use the notation
\begin{gather*}
\zeta=\psi_{xx}+\psi_{yy},\quad
\theta=\psi_{xx}-\psi_{yy},\quad
\eta=2\psi_{xy},\\
\sigma=\psi_{xxx}-3\psi_{xyy},\quad
\varsigma=3\psi_{xxy}-\psi_{yyy},\quad
V=\mathrm D_t+\psi_x\mathrm D_y-\psi_y\mathrm D_x.
\end{gather*}
The algebra~$\mathfrak g_0$ possesses no differential invariants up to order two.
At the same time, it has the singular second order manifold determined by the equations $\psi_{xx}=\psi_{yy}$ and $\psi_{xy}=0$,
which is not essential for our consideration.
A generating set~$\mathcal I_0$ of functionally independent differential invariants of~$\mathfrak g_0$ consists of the third order differential invariants
\begin{gather*}
\frac{V\zeta}{\theta^2+\eta^2},\quad
\frac{\theta V\theta+\eta V\eta}{(\theta^2+\eta^2)^{3/2}},\quad
\frac{(V\theta+2\eta\zeta)^2+(V\eta-2\theta\zeta)^2}{(\theta^2+\eta^2)^2},\\
\frac{\sigma^2+\varsigma^2}{\zeta_x{}^2+\zeta_y{}^2},\quad
\frac{\theta(\zeta_x{}^2-\zeta_y{}^2)+2\eta\zeta_x\zeta_y}{(\theta^2+\eta^2)^{1/2}(\zeta_x{}^2+\zeta_y{}^2)},\quad
\frac{\sigma\zeta_x(\zeta_x{}^2-3\zeta_y{}^2)+\varsigma\zeta_y(3\zeta_x{}^2-\zeta_y{}^2)}{(\zeta_x{}^2+\zeta_y{}^2)^{3/2}}.
\end{gather*}
A complete set~$\mathcal O_0$ of independent operators of invariant differentiation for this algebra is formed by the operators
\[
(\theta^2+\eta^2)^{-1/2}V,\quad
(\zeta_x{}^2+\zeta_y{}^2)^{-1/2}(\zeta_x\mathrm D_x+\zeta_y\mathrm D_y),\quad
(\zeta_x{}^2+\zeta_y{}^2)^{-1/2}(\zeta_x\mathrm D_y-\zeta_y\mathrm D_x).
\]

The computation of~$\mathcal I_0$ and~$\mathcal O_0$ is cumbersome and will be presented elsewhere, jointly with the selection of a basis (i.e., minimal generating set) of differential invariants.
At the same time, the result of the computation can be checked in a rather direct and simple way.
Indeed, the cardinality of~$\mathcal O_0$ equals three.
The elements of~$\mathcal O_0$ are linearly independent over the ring of differential invariants of~$\mathfrak g_0$ and
commute with the infinite prolongations of all vector fields from the generating set~\eqref{eq:algebravort} of~$\mathfrak g_0$.
Since each element~$I$ of~$\mathcal I_0$ satisfies the condition $Q_{(3)}I=0$, where the operator~$Q$ runs through the operators~\eqref{eq:algebravort}, it is a differential invariant of~$\mathfrak g_0$.
The invariants belonging to~$\mathcal I_0$ are functionally independent.
Moreover, for any fixed order~$r$ an $r$th order universal basis of differential invariants of~$\mathfrak g_0$ can be constructed via acting by  operators from~$\mathcal O_0$ on invariants from~$\mathcal I_0$. We only sketch the proof of the last assertion.
The cardinality of any $r$th order universal basis of differential invariants of~$\mathfrak g_0$, where $r\geqslant4$, equals
the difference between the dimension of the jet space~$J^r$ and the rank of the $r$th prolongation of~$\mathfrak g_0$,
\[
N=3+\binom{r+3}{r}-(3r+8). 
\]
Acting on elements of~$\mathcal I_0$ by operators from~$\mathcal O_0$ $k-3$ times, $3\leqslant k\leqslant r$,
we obtain a set of $k$-order differential invariants which is of maximal rank with respect to the $k$-order derivatives
involving at least two differentiations with respect to space variables.
Choosing, for each~$k$, a subset of invariants associated with a nonzero $k$-order minor in the corresponding Jacobi matrix
and uniting such subsets for $k\leqslant r$, we construct exactly $N$ functionally independent differential invariants
of order not greater than~$r$, which hence form an $r$th order universal basis of differential invariants of~$\mathfrak g_0$.

The above case where invariance of the parameterization under the whole symmetry group of the vorticity equation is desired can be neglected for physical reasons. This is since it is impossible to realized, e.g.\ the down-gradient ansatz within this framework. It can easily be checked that the corresponding vorticity equation with parameterized eddy vorticity flux only admits one scaling operator for any physically meaningful ansatz for $K$. In contrast to the example of the Navier--Stokes equations discussed in~\cite{ober97Ay}, the vorticity equation hence does not allow physical parameterizations leading to a closed model invariant under the same symmetry group as the original vorticity equation. This is why it is beneficial to single out several subgroups of the maximal Lie invariance group and consider the invariant parameterization problem only with respect to these subgroups.

\medskip
\noindent\textbf{Explicit spatial dependency.} If the two-dimensional fluid is anisotropic and inhomogeneous the only subalgebra of~\eqref{eq:algebravort} that can be admitted is spanned by the operators
\[
    \p_t, \quad \ZZ(\chi) = \chi(t)\p_{\psi}.
\]
For this subalgebra, a basis of invariants is formed by $x$, $y$, $\psi_x$ and $\psi_y$.
Independent operators of invariant differentiation are exhausted by $\mathrm D_t$, $\mathrm D_x$ and $\mathrm D_y$.
If we express the right hand side of~\eqref{eq:vorta} in terms of differential invariants of the above subalgebra, a possible representation reads
\[
     \zeta_t + \{ \psi, \zeta\} = K(x,y)\nn^2 \zeta.
\]
Hence we assembled our parameterization using the (differential) invariants $x$, $y$, $\mathrm D_x^3\psi_x=\psi_{xxxx}$, $\mathrm D_y^2\mathrm D_x\psi_x= \psi_{xxyy}$ and $\mathrm D_y^3\psi_y=\psi_{yyyy}$. This boils down to the usual gradient ansatz for the eddy flux term, where the eddy viscosity $K$ explicitly depends on the position in the space. Note, however, that this ansatz is only one possibility which is feasible within this class of models.

\medskip

\noindent\textbf{Rotationally invariant fluid.} In case the two-dimensional fluid is isotropic, the resulting parameterized system should also admit rotations. Hence, we seek for differential invariants of the subalgebra~$\mathfrak j$ spanned by the operators
\[
    \p_t, \quad \JJ = x\p_y-y\p_x, \quad \JJ^t = tx\p_y-ty\p_x + \tfrac12(x^2+y^2)\p_{\psi},\quad \ZZ(\chi) = \chi(t)\p_{\psi}.
\]
A complete set of independent operators of invariant differentiation for~$\mathfrak j$ consists of
\[
\mathrm D_t+\psi_x\mathrm D_y-\psi_y\mathrm D_x,\quad
x\mathrm D_x+y\mathrm D_y,\quad
-y\mathrm D_x+x\mathrm D_y
\]
and a generating set of functionally independent differential invariants is formed by
\begin{gather*}
\rho=\tfrac12(x^2+y^2),\quad x\psi_y-y\psi_x,\quad (x^2+y^2)(\psi_{xx}+\psi_{yy})-2(x\psi_x+y\psi_y),\\
(x^2+y^2)(x\psi_{tx}+y\psi_{ty})+(x\psi_x+y\psi_y)(x y(\psi_{yy}-\psi_{xx})+(x^2-y^2)\psi_{xy}+x\psi_y-y\psi_x).
\end{gather*}
In the modified polar coordinates $(\rho,\varphi)$ with $\varphi=\arctan y/x$, these sets have, after an additional rearrangement, simpler representations
$\mathcal O=\{\tilde{\mathrm D}_t,\mathrm D_\rho,\mathrm D_\varphi\}$ and
$\mathcal I=\{I^\alpha,\,\alpha=0,\dots,3\}$, respectively,
where $\tilde{\mathrm D}_t=\mathrm D_t+\psi_\rho\mathrm D_\varphi$ and
\[
I^0=\rho,\quad I^1=\psi_\varphi,\quad I^2=\psi_{\rho\rho},\quad I^3=\psi_{t\rho}+\psi_\rho\psi_{\rho\varphi}.
\]

Any element of~$\mathcal O$ indeed is an invariant differentiation operator for~$\mathfrak j$
since it commutes with the infinite prolongation of every vector field from~$\mathfrak j$.
The fact that $I^\alpha$, $\alpha=0,\dots,3$, are differential invariants of~$\mathfrak j$ is also checked in a rather direct way,
by the substitution to the condition $Q_{(2)}I=0$, where the operator~$Q$ runs through~$\mathfrak j$.
These invariants obviously are functionally independent.

The most difficult part is to prove that for any fixed order~$r$ we can construct an $r$th order universal basis of differential invariants of~$\mathfrak j$ by invariant differentiations of $I^\alpha$, $\alpha=0,\dots,3$.
The number of elements in any $r$th order universal basis of differential invariants of~$\mathfrak j$, where $r\geqslant1$, equals
\[
3+\binom{r+3}{r}-(r+4)=\frac r6(r+1)(r+5)
\]
(the dimension of the jet space~$J^r$ minus the rank of the $r$th prolongation of~$\mathfrak j$).
The commutation relations between the operators of invariant differentiation are
\[
[\mathrm D_\rho,\mathrm D_\varphi]=0,\quad
[\mathrm D_\rho,\tilde{\mathrm D}_t]=\psi_{\rho\rho}\mathrm D_\varphi,\quad
[\mathrm D_\varphi,\tilde{\mathrm D}_t]=\psi_{\rho\varphi}\mathrm D_\varphi.
\]
The elements $I^1$, $I^2$ and $I^3$ of~$\mathcal I$ can be represented in the form
$I^1=\mathrm D_\varphi\psi$, $I^2=\mathrm D_\rho^2\psi$ and $\smash{I^3=\tilde{\mathrm D}_t\mathrm D_\rho\psi}$.
Hence, acting by the operators of invariant differentiation on elements of~$\mathcal I$, we can construct
\[
1+\binom{r-1+3}{r-1}+\binom{r-2+2}{r-2}+\binom{r-2+1}{r-2}=\frac r6(r+1)(r+5)
\]
functionally independent invariants of order not greater than~$r$
(the zeroth order invariant~$\rho$ plus
acting on~$I^1$ by the operators $\tilde{\mathrm D}_t^{\alpha_1}\mathrm D_\rho^{\alpha_2}\mathrm D_\varphi^{\alpha_3}$, where $\alpha_1+\alpha_2+\alpha_3\leqslant r-1$, plus
acting on~$I^2$ by $\mathrm D_\rho$ and then $\tilde{\mathrm D}_t$ at most $r-2$ times in total and plus
acting on~$I^3$ by $\mathrm D_\rho$ at most $r-2$ times).
As the above numbers coincide, the proof is completed.

For the set~$\mathcal O$ of operators of invariant differentiation,
the generating set~$\mathcal I$ of differential invariants is not minimal.
On the domain singled out in the corresponding infinite jet space by the condition $\psi_{\varphi\varphi}\ne0$
we have
\[
I^2=\frac{[\mathrm D_\rho,\tilde{\mathrm D}_t]I^1}{\mathrm D_\varphi I^1}
\]
and hence the invariant~$I^2$ can be excluded from the generating set of invariants.
At the same time, the remaining invariant $I^0$, $I^1$ and $I^3$ form a basis (i.e., minimal generating set) of invariants with respect to the set~$\mathcal O$ of operators of invariant differentiation.
Indeed, any function of~$\rho$ and invariants obtained from~$I^1$ by invariant differentiations is represented as
a function of~$\rho$, $\tilde{\mathrm D}_t^{\alpha_1}\mathrm D_\rho^{\alpha_2}\mathrm D_\varphi^{\alpha_3}\psi_\varphi$ and $\tilde{\mathrm D}_t^{\beta_1}\mathrm D_\rho^{\beta_2}\psi_{\rho\rho}$,
where $(\alpha_1,\alpha_2,\alpha_3)$ and $(\beta_1,\beta_2)$ run through certain subsets of $\mathbb N_0^3$ and $\mathbb N_0^2$, respectively,
and hence this function cannot coincide with $I^3$.
Analogously, any function of~$\rho$ and invariants constructed from~$I^3$ by invariant differentiations is represented as
a function of~$\rho$ and $\tilde{\mathrm D}_t^{\alpha_1}\mathrm D_\rho^{\alpha_2}\mathrm D_\varphi^{\alpha_3}\psi_\rho$,
where $(\alpha_1,\alpha_2,\alpha_3)$ runs through certain subset of $\mathbb N_0^3$,
and hence this function cannot coincide with $I^1$.

As an example, the parameterizations of the form
\[
     \zeta_t + \{\psi, \zeta\} = K(\sqrt{x^2+y^2}\,)\nn^2\zeta
\]
are invariant with respect to~$\mathfrak j$ because $\rho$, $\zeta_t + \{\psi, \zeta\}$ and $\nn^2 \zeta$ are differential invariants of~$\mathfrak j$.

\medskip

In the same fashion it would be possible to derive classes of parameterizations that preserve other subalgebras of $\mathfrak g_0$, e.g.\ including (generalized) Galilean symmetry or a scaling symmetry, but we do not derive them in this paper.

\subsection{Equivalence algebras of classes of generalized vorticity equations}\label{sec:EquivAlgebra}

In order to demonstrate different possible techniques, we present the details of the calculation of the usual equivalence algebra~$\mathfrak g^\sim_1$ for the class of equations
\begin{equation}\label{eq:vortdirect}
    \zeta_t + \{\psi,\zeta\} = \mathrm D_i f^i(t,x,y,\zeta_x,\zeta_y) =  f^i_i + f^i_{\zeta_j}\zeta_{ij}, \quad
    \zeta := \psi_{ii},
\end{equation}
where for convenience we introduce another notation for the independent variables, $t=z_0$, $x=z_1$ and $y=z_2$, and omit bars over the dependent variables.
Throughout the section the indices $i$, $j$ and~$k$ range from $1$ to $2$, while the indices~$\kappa$, $\lambda$, $\mu$ and $\nu$ run from $0$ to $2$.
The summation over repeated indices is understood.
A numerical subscript of a function denotes the differentiation with respect to the corresponding variable~$z_\mu$.

In fact, the equivalence algebra of class~\eqref{eq:vortdirect} can be easily obtained from the much more general results on admissible transformations, presented in Section~\ref{sec:AdmTrans}.
At the same time, calculations using the direct method applied for finding admissible transformations are too complicated and lead to solving nonlinear overdetermined systems of partial differential equations.
This is why the infinitesimal approach is wider applied and realized within symbolic calculation systems.
The usage of the infinitesimal approach for the construction of the equivalence algebra of~\eqref{eq:vortdirect} has specific features richly deserving to be demonstrated here.

\begin{theorem}\label{TheoremOnEquivAlgebra1}
The equivalence algebra~$\mathfrak g^\sim_1$ of class~\eqref{eq:vortdirect} is generated by the operators
\begin{align}\label{eq:equivalgebravort}
\begin{split}
&\tilde\DDD_1 = t\p_t-\psi\p_{\psi}-\zeta_x\p_{\zeta_x}-\zeta_y\p_{\zeta_y}-2f^1\p_{f^1}-2f^2\p_{f^2},\qquad \p_t,\\[.5ex]
&\tilde\DDD_2= x\p_x + y\p_y + 2\psi\p_{\psi} -\zeta_x\p_{\zeta_x}-\zeta_y\p_{\zeta_y} + f^1\p_{f^1} + f^2\p_{f^2}, \\
&\tilde\JJ(\beta)=\beta x\p_y-\beta y\p_x+\frac{\beta_t}{2}(x^2+y^2)\p_\psi+\beta(\zeta_x\p_{\zeta_y}-\zeta_y\p_{\zeta_x})\\
&\qquad\quad\ {}+(\beta_{tt}x-\beta f^2)\p_{f^1}+(\beta_{tt}y+\beta f^1)\p_{f^2}, \\[.5ex]
&\tilde\XX(\gamma^1)=\gamma^1\p_x-\gamma_t^1y\p_\psi,\qquad
 \tilde\YY(\gamma^2)=\gamma^2\p_y+\gamma^2_tx\p_\psi,\\
&\tilde{\mathcal R}(\sigma) = \frac\sigma2(x^2+y^2)(\p_\psi+\zeta_y\p_{f^1}-\zeta_x\p_{f^2})+\sigma_tx\p_{f^1}+\sigma_ty\p_{f^2},\\
&\tilde{\mathcal H}(\delta)=\delta(\p_\psi +\zeta_y\p_{f^1}-\zeta_x\p_{f^2}),\qquad
 \tilde{\mathcal G}(\rho)=\rho_x\p_{f^2}-\rho_y\p_{f^1},\qquad
 \tilde{\mathcal Z}(\chi)=\chi\p_\psi,
\end{split}
\end{align}
where $\beta$, $\gamma^i$, $\sigma$ and $\chi$ are arbitrary smooth functions of $t$ solely, $\delta=\delta(t,x,y)$ is an arbitrary solution of the Laplace equation $\delta_{xx}+\delta_{yy}=0$ and $\rho=\rho(t,x,y)$ is an arbitrary smooth function of its arguments.
\end{theorem}

\begin{remark*}
Although the coefficients of~$\p_{\zeta_x}$ and~$\p_{\zeta_y}$ can be obtained by standard prolongation from the coefficients associated with the equation variables, it is necessary to include the corresponding terms in the representation of the basis elements~\eqref{eq:equivalgebravort} in order to guarantee that they commute in a proper way.
\end{remark*}

\begin{remark*}
The operators $\tilde{\mathcal G}(\rho)$ and $\tilde{\mathcal H}(\chi)-\tilde{\mathcal Z}(\chi)$ arise due to the total divergence expression of the right hand side of the first equation in~\eqref{eq:vortdirect}, leading to the gauge freedom in rewriting the right hand side of the class~\eqref{eq:vortdirect}. They do not generate transformations of the independent and dependent variables and hence form the gauge equivalence subalgebra of the equivalence algebra~\eqref{eq:equivalgebravort} \cite{popo10Ay}. The parameter-function~$\rho$ is defined up to summand depending on~$t$.
\end{remark*}

\begin{proof}
As coordinates in the underlying fourth-order jet space~$\mathfrak J^{(4)}$, we choose the variables
\begin{gather*}
z_\mu,\quad \psi,\quad \psi_\mu,\quad \psi_{\mu\nu},\ \mu\leqslant\nu,\quad
\psi_{\lambda\mu\nu},\ \lambda\leqslant\mu\leqslant\nu,\ (\mu,\nu)\ne(2,2),\quad \zeta_\mu,\\
\psi_{\kappa\lambda\mu\nu},\ \kappa\leqslant\lambda\leqslant\mu\leqslant\nu,\ (\mu,\nu)\ne(2,2),\quad \zeta_{\mu\nu},\ \mu\leqslant\nu.
\end{gather*}
(Variables of the jet space and related values are defined by their notation up to permutation of indices.)
The variable $\zeta_0$ of the jet space is assumed principal, i.e., it is expressed via the other coordinate variables (called the parametric ones) in view of equation~\eqref{eq:vortdirect}. Under calculation we also carry out the substitutions $\psi_{22\mu}=\zeta_\mu-\psi_{11\mu}$.
To avoid repetition of the above conditions for indices, in what follows we assume
that the index tuples $(\mu,\nu)$, $(\lambda,\mu,\nu)$ and $(\kappa,\lambda,\mu,\nu)$ satisfy these conditions by default.

Due to the special form of the arbitrary elements $f^i$, we have to augment equation~\eqref{eq:vortdirect} with the following auxiliary system for~$f^i$:
\begin{equation}\label{eqb}
    f^i_\psi = f^i_{\psi_\mu} = f^i_{\psi_{\mu\nu}} = f^i_{\psi_{\lambda\mu\nu}} = f^i_{\zeta_0} = f^i_{\psi_{\kappa\lambda\mu\nu}} = f^i_{\zeta_{\mu\nu}} = 0.
\end{equation}

As we compute the usual equivalence algebra rather than the generalized one \cite{mele96Ay}
and the arbitrary elements $f^i$ do not depend on fourth order derivatives of~$\psi$,
the elements of the algebra are assumed to be vector fields in the joint space of the variables of~$\mathfrak J^{(3)}$ and the arbitrary elements $f^i$,
which are projectable to both the spaces $(t,x,y,\psi)$ and $\mathfrak J^{(3)}$.
In other words, the algebra consists of vector fields of the general form
\[
Q=\xi^\mu\p_\mu+\eta\p_\psi+\eta^\mu\p_{\psi_\mu}+\eta^{\mu\nu}\p_{\psi_{\mu\nu}}
+\eta^{\lambda\mu\nu}\p_{\psi_{\lambda\mu\nu}}+\theta^\mu\p_{\zeta_\mu}+\varphi^i\p_{f^i},
\]
where $\xi^\mu=\xi^\mu(t,x,y,\psi)$, $\eta=\eta(t,x,y,\psi)$,
the coefficients corresponding to derivatives of~$\psi$ are obtained by the standard prolongation~\eqref{EqProlongedOp} from~$\xi^\mu$ and~$\eta$, the coefficients~$\theta^\nu$  are obtained by the standard prolongation from~$\xi^\mu$ and~$\theta=\eta^{ii}$,
and the coefficients $\varphi^i$ depends on all the variables of~$\mathfrak J^{(3)}$ and the arbitrary elements $f^j$.
As a result, each element from the equivalence algebra is determined by its coefficients~$\xi^\mu$, $\eta$ and~$\varphi^i$.
To act on the equations~\eqref{eq:vortdirect} and~\eqref{eqb} by the operator~$Q$, we should additionally prolong it
to the variables $\psi_{\kappa\lambda\mu\nu}$ and $\zeta_{\mu\nu}$ in the conventional way and
to the derivatives of~$f$, assuming all the variables of~$\mathfrak J^{(3)}$ as usual ones:
\begin{align*}
\bar Q={}&Q+\eta^{\kappa\lambda\mu\nu}\p_{\psi_{\kappa\lambda\mu\nu}}+\theta^{\mu\nu}\p_{\zeta_{\mu\nu}}\\&
+\varphi^{i\mu}\p_{f^i_\mu}+\varphi^{i\psi}\p_{f^i_\psi}+\varphi^{i\psi_\mu}\p_{f^i_{\psi_\mu}}
+\varphi^{i\psi_{\mu\nu}}\p_{f^i_{\psi_{\mu\nu}}}
+\varphi^{i\psi_{\lambda\mu\nu}}\p_{f^i_{\psi_{\lambda\mu\nu}}}+\varphi^{i\zeta_\mu}\p_{f^i_{\zeta_\mu}}.
\end{align*}

First we consider the infinitesimal invariance conditions associated with equations~\eqref{eqb}.
The invariance condition for the equation $f^i_\psi=0$ is
\[
    \varphi^{i\psi}\big|_{\rm Eq.~\eqref{eqb}} = \varphi^i_\psi  -\xi^\mu_\psi f^i_\mu - \theta^k_\psi f^i_{\zeta_k}=0.
\]
Splitting with respect to derivatives of $f^i$ in the latter equation implies that $\varphi^i_\psi=0$, $\xi^\mu_\psi = 0$, $\theta^i_\psi = 0$.
As $\theta^i=\mathrm D_j\mathrm D_j\mathrm D_i(\eta-\xi^\mu\psi_\mu)+\xi^\mu\psi_{\mu jji}$,
we additionally derive the simple determining equation $\eta_{\psi\psi}=0$.

In a similar way, the invariance conditions for the equations
$f^i_{\psi_\mu}=0$, $f^i_{\psi_{\mu\nu}}=0$, $f^i_{\psi_{\lambda\mu\nu}}=0$ and $f^i_{\zeta_0} = 0$
can be presented in the form
\begin{gather*}
    \varphi^{i\psi_{\mu}}\big|_{\rm Eq.~\eqref{eqb}} = \varphi^i_{\psi_\mu} - \theta^k_{\psi_\mu} f^i_{\zeta_k} = 0,\\
    \varphi^{i\psi_{\mu\nu}}\big|_{\rm Eq.~\eqref{eqb}} = \varphi^i_{\psi_{\mu\nu}} - \theta^k_{\psi_{\mu\nu}} f^i_{\zeta_k} = 0,\\
    \varphi^{i\psi_{\lambda\mu\nu}}\big|_{\rm Eq.~\eqref{eqb}} = \varphi^i_{\psi_{\lambda\mu\nu}} - \theta^k_{\psi_{\lambda\mu\nu}} f^i_{\zeta_k}=0,\\
    \varphi^{i\zeta_0}\big|_{\rm Eq.~\eqref{eqb}} = \varphi^i_{\zeta_0} - \theta^k_{\zeta_0} f^i_{\zeta_k} = 0,
\end{gather*}
which is split into
$\varphi^i_{\psi_\mu}=0$, $\theta^k_{\psi_\mu}=0$;
$\varphi^i_{\psi_{\mu\nu}}=0$, $\theta^k_{\psi_{\mu\nu}}=0$;
$\varphi^i_{\psi_{\lambda\mu\nu}}=0$, $\theta^k_{\psi_{\lambda\mu\nu}}=0$; and
$\varphi^i_{\zeta_0}=0$, $\theta^k_{\zeta_0}=0$, respectively.
The equations $\theta^k_{\psi_\mu}=0$, $\theta^k_{\psi_{\lambda\mu\nu}}=0$ and $\theta^k_{\zeta_0}=0$ provide no essential restrictions on the coefficients~$\xi^\mu$, $\eta$ and~$\varphi^i$. From the equation $\theta^k_{\psi_{\lambda\mu\nu}}=0$ we derive that $\xi^0_j=0$, $\xi^1_2+\xi^2_1=0$ and $\xi^1_1-\xi^2_2=0$.
Hence
\[
    \theta = \eta^{jj}= \eta_{jj} + 2\eta_{j\psi}\psi_j + \eta_\psi\psi_{jj} - 2\xi^i_j\psi_{ij}= \eta_{jj} + 2\eta_{j\psi}\psi_j + (\eta_\psi-2\xi^1_1)\zeta.
\]

It remains to solve the determining equations following from the invariance condition for equation~\eqref{eq:vortdirect}. The invariance condition reads
\[
    \theta^0 + \eta^1\zeta_2 + \psi_1\theta^2 - \eta^2\zeta_1 + \psi_2\theta^1 = \varphi^{ii} + \varphi^{i\zeta_j}\zeta_{ji} + f^i_{\zeta_j}\theta^{ji},
\]
or explicitly
\begin{gather*}
\eta_{jjt} + \eta_{jj\psi}\psi_t + 2\eta_{tj\psi}\psi_j + 2\eta_{j\psi}{\psi_{tj}} + (\eta_{t\psi}-2\xi^1_{t1})\zeta\\
+ (\eta_\psi-2\xi^1_1-\xi^0_t)(f^i_i + f^i_{\zeta_j}\zeta_{ij} - \psi_1\zeta_2+\psi_2\zeta_1) - \xi^i_t\zeta_i \\
+ (\eta_1+\eta_\psi\psi_1 - \xi^i_1\psi_i)\zeta_2 + \psi_1(\eta_{jj2} + \eta_{jj\psi}\psi_2+2\eta_{2j\psi}\psi_j + 2\eta_{j\psi}\psi_{2i}+ (\eta_\psi-2\xi^1_1)\zeta_2 - \xi^i_2\zeta_i)\\
- (\eta_2+\eta_\psi\psi_2 - \xi^i_2\psi_i)\zeta_1 - \psi_2(\eta_{jj1} + \eta_{jj\psi}\psi_1 + 2\eta_{1j\psi}\psi_j + 2\eta_{j\psi}\psi_{1j} + (\eta_\psi-2\xi^1_1)\zeta_1 - \xi^i_1\zeta_i) \\
=\varphi^i_i + \varphi^i_{f^j}f^j_i - \xi^j_if^i_j - \theta^k_i f^i_{\zeta_k} + \zeta_{ij}(\varphi^i_{\zeta_j} + \varphi^i_{f^k}f^k_{\zeta_j}-\theta^k_{\zeta_j}f^i_{\zeta_k}) + f^i_{\zeta_j}\theta^{ij}.
\end{gather*}
Collecting the coefficients of $\psi_{tj}$ gives $\eta_{j\psi} = 0$. This implies that $\theta_\psi = 0$.
Similarly, the coefficients of $\psi_i\zeta_j$ lead to the equation
$\eta_{ijj}=0$ and $\eta_\psi-2\xi^1_1 + \xi^0_t=0$.
As $\xi^0_i=0$ and $\eta_{i\psi}=0$, the second equation together with the relations $\xi^1_1=\xi^2_2$ and $\xi^1_2+\xi^2_1=0$ implies that $\xi^i_{jk}=0$.
Then, the coefficient of $\zeta$ gives $\xi^0_{tt}=0$ and the coefficients of $f^i_j$ lead to $\varphi^1_{f^2}= \xi^1_2$, $\varphi^2_{f^1}= \xi^2_1$ and  $\varphi^1_{f^1}=\varphi^2_{f^2}=\xi^1_1-2\xi^0_t$.
In view of the determining equations that we have already derived, the terms involving $f^i_{\zeta_j}$ are identically canceled.
Note that the coefficients of $f^i_{\zeta_j}\zeta_{kl}$ simultaneously lead to the same set of equations as the coefficients of $f^i_j$.

The remaining part of the invariance condition is
$\eta_{jjt} - \xi^i_t\zeta_i + \eta_1\zeta_2-\eta_2\zeta_1 = \varphi^i_i +\zeta_{ij}\varphi^i_{\zeta_j}$.
Splitting with respect to $\zeta_{ij}$ in this relation gives $\varphi^1_{\zeta_1}=\varphi^2_{\zeta_2}=0$, $\varphi^1_{\zeta_2}+\varphi^2_{\zeta_1}=0$ and
\[
    \varphi^i_i=\eta_{jjt} - \xi^i_t\zeta_i + \eta_1\zeta_2-\eta_2\zeta_1.
\]
Acting on the last equation by the operator $\p_j\p_{\zeta_j}$, we obtain $\xi^i_{it}=0$.
Further splitting with respect to $\zeta_1$ and $\zeta_2$ is not possible since $\varphi^j$ may depend on them.

Finally, the reduced system of determining equations reads
\begin{gather*}
    \xi^0_\psi=\xi^0_i=\xi^0_{tt}=0, \quad \xi^i_\psi=\xi^i_{jk}=0, \quad \xi^i_{it}=0, \quad \xi^1_1=\xi^2_2, \quad \xi^1_2+\xi^2_1 = 0,\\
    \eta_{\psi\psi}=0,\quad \eta_\psi=2\xi^1_1-\xi^0_t,\quad \eta_{ijj}=0, \\
    \varphi^i_\psi=0, \quad \varphi^i_{\psi_\mu}=0, \quad \varphi^i_{\psi_{\mu\nu}}=0, \quad \varphi^i_{\psi_{\lambda\mu\nu}}=0, \quad \varphi^i_{\zeta_0}=0,\\
    \varphi^1_{f^2}= \xi^1_2, \quad \varphi^2_{f^1}= \xi^2_1, \quad \varphi^1_{f^1}=\varphi^2_{f^2}=\xi^1_1-2\xi^0_t,\\
    \varphi^1_{\zeta_1}=\varphi^2_{\zeta_2}= 0, \quad \varphi^1_{\zeta_2}+\varphi^2_{\zeta_1} = 0, \quad
    \varphi^i_i=\eta_{jjt} - \xi^i_t\zeta_i + \eta_1\zeta_2-\eta_2\zeta_1.
\end{gather*}
The solution of this system provides the principal coefficients of the operators from the equivalence algebra of the class~\eqref{eq:vortdirect}:
\begin{align}\label{eq:infequiv}
\begin{split}
    &\xi^0 = c_1t+c_0,\qquad \xi^1=c_2 x-\beta y+\gamma^1, \qquad \xi^2=\beta x+c_2 y+\gamma^2,  \\
    &\eta = (2c_2-c_1)\psi + \delta -\gamma^1_ty + \gamma^2_tx+\frac{\beta_t}2(x^2+y^2)+\frac\sigma2(x^2+y^2)+\chi, \\
    &\varphi^1= (c_2-2c_1)f^1-\beta f^2 + \delta\zeta_y+\frac\sigma2(x^2+y^2)\zeta_y+\beta_{tt}x+\sigma_t x-\rho_y,\\
    &\varphi^2= \beta f^1+(c_2 -2c_1)f^2 -\delta\zeta_x-\frac\sigma2(x^2+y^2)\zeta_x+\beta_{tt}y+\sigma_t y+\rho_x,
\end{split}
\end{align}
where $\beta$, $\gamma^i$, $\sigma$ and $\chi$ are real-valued smooth functions of~$t$ only,
$c_0$, $c_1$ and $c_2$ are arbitrary constants, $\rho$ is an arbitrary function of~$t$, $x$ and~$y$ and
$\delta=\delta(t,x,y)$ is an arbitrary solution of the Laplace equation $\delta_{jj} = 0$.

Splitting with respect to parametric values in~\eqref{eq:infequiv}, we obtain the coefficients of the basis operators~\eqref{eq:equivalgebravort}
of the algebra~$\mathfrak g^\sim_1$.
Recall that the coefficients $\eta^\mu$, $\eta^{\mu\nu}$, $\eta^{\lambda\mu\nu}$ and~$\theta^\nu$ are calculated from~$\xi^\mu$ and~$\eta$
via the standard procedure of prolongation and the coefficients~$\varphi^i$ do not depend on
$\psi_\mu$, $\psi_{\mu\nu}$, $\psi_{\lambda\mu\nu}$, and $\zeta_0$.
Therefore, both the operators from~$\mathfrak g^\sim_1$ and their commutators are completely determined
by the coefficients of $\p_{\mu}$, $\p_\psi$, $\p_{\zeta_i}$ and $\p_{f^j}$.
This is why in~\eqref{eq:equivalgebravort} and similar formulas we omit the other terms for sake of brevity.
\end{proof}

\begin{remark*}\looseness=-1
The auxiliary system for the arbitrary elements is an important component of the definition of a class of differential equations.
Its choice is usually guided by some prior knowledge about the processes to be parameterized.
We have decided to assume that the arbitrary elements~$f^1$ and~$f^2$ depend also on $t$, keeping in mind two more, purely mathematical, reasons.
The first reason is that the projection of the corresponding equivalence algebra on the space $(t,x,y,\psi)$ contains the maximal Lie invariance algebra~$\mathfrak g_0$ of the vorticity equation~\eqref{eq:vort} which is the initial point of the entire consideration.
The basis operators~\eqref{eq:algebravort} of~$\mathfrak g_0$ are obtained from~\eqref{eq:equivalgebravort}~by
\begin{gather*}
\DDD_1 = \pr\tilde \DDD_1, \qquad \p_t = \pr \p_t, \qquad \DDD_2 = \pr \tilde\DDD_2, \qquad \JJ = \pr\tilde\JJ(1), \qquad \JJ^t = \pr\tilde\JJ(t), \\
\XX(\gamma^1) = \pr\tilde\XX(\gamma^1), \qquad \YY(\gamma^2) = \pr\tilde\YY(\gamma^2),\qquad \ZZ(\chi) = \pr\tilde\ZZ(\chi),
\end{gather*}
where $\pr$ denotes the projection operator on the space $(t,x,y,\psi)$.
(Though the expressions for the operator $\p_t$ (resp.\ $\tilde\XX(\gamma^1)$, $\tilde\YY(\gamma^2)$ or $\tilde\ZZ(\chi)$) and its projection formally coincide, they in fact determine vector fields on different spaces.)
The second reason is that the class~\eqref{eq:vortdirect} is normalized, cf. Section~\ref{sec:AdmTrans}.
This in particular implies that the maximal Lie invariance algebra of any equation from the class~\eqref{eq:vortdirect} is contained in the projection of the equivalence algebra~$\mathfrak g^\sim_1$ of this class.
\end{remark*}

We also calculate the equivalence algebras of two subclasses of the class~\eqref{eq:vortdirect}.

The first subclass corresponds to parameterizations not depending on time explicitly and, therefore, is singled out from the class~\eqref{eq:vortdirect}
by the further auxiliary equation
\[
f^i_t=0,
\]
which has no influence on splitting of the invariance conditions for the equations~\eqref{eq:vortdirect} and~\eqref{eqb} and
gives the additional determining equations $\varphi^i_t=\xi^i_t=\theta^i_t=0$.
These determining equations imply that $\beta$, $\gamma^i$ and $\sigma$ are constant, $\delta$ is a function only of~$x$ and~$y$
and $\rho$ can be assumed as a function only of~$x$ and~$y$.
Therefore, the equivalence algebra of this subclass is
\[
\langle\tilde\DDD_1,\,\p_t,\,\tilde \DDD_2,\,\tilde\JJ(1),\,\tilde\XX(1),\,\tilde\YY(1),\,\tilde{\mathcal R}(1),\,\tilde{\mathcal H}(\delta),\,
\tilde{\mathcal G}(\rho),\,\tilde{\mathcal Z}(\chi)\rangle,
\]
where the parameter-function $\delta=\delta(x,y)$ runs through the set of solutions of the Laplace equation $\delta_{xx}+\delta_{yy}=0$ and
$\rho=\rho(x,y)$ is an arbitrary function of its arguments.

The second subclass is associated with spatially independent parameterizations. Hence we additionally set
\[
f^i_j = 0.
\]
It has to be noted that after attaching this condition we cannot split with respect to $f^i_j$ as we did in the course of solving the determining equations. However, precisely the same conditions obtained from splitting with respect to $f^i_j$ can also be obtained from splitting with respect to $f^i_{\zeta_j}$. Hence the condition $f^i_j=0$ only leads to the additional restriction $\varphi^i_j=0$ and, therefore, we find that $\delta_i=0$, $\sigma=0$, $\beta_{tt}=0$ and $\rho_{ij}=0$.
Without loss of generality we can set $\rho=\rho^i(t)z_i$, where $\rho^i$ are arbitrary smooth functions of~$t$.
As a result, the equivalence algebra~$\mathfrak g^\sim_2$ of the second subclass is generated by the operators
\[
\tilde\DDD_1,\,\p_t,\,\tilde \DDD_2,\,\tilde\JJ(1),\,\tilde\JJ(t),\,\tilde\XX(\gamma^1),\,\tilde\YY(\gamma^2),\,\tilde{\mathcal H}(\delta),\,
\tilde{\mathcal G}(\rho^1x+\rho^2y),\,\tilde{\mathcal Z}(\chi),
\]
where $\gamma^i$, $\rho^i$, $\delta$ and~$\chi$ are arbitrary smooth functions of~$t$.

The intersection of the above subclasses corresponds to the set of parameterizations independent of both $t$ and $(x,y)$ and
is singled out from the class~\eqref{eq:vortdirect} by the joint auxiliary system \[f^i_t=f^i_j=0.\]
Its equivalence algebra is the intersection of the equivalence algebras of the above subclasses and, therefore, equals
\[
\langle\tilde\DDD_1,\,\p_t,\,\tilde \DDD_2,\,\tilde\JJ(1),\,\tilde\XX(1),\,\tilde\YY(1),\,\tilde{\mathcal H}(1),\,
\tilde{\mathcal G}(\rho^1x+\rho^2y),\,\tilde{\mathcal Z}(\chi)\rangle,
\]
where $\rho^1$, $\rho^2$ and~$\chi$ are arbitrary smooth functions of~$t$.

\subsection{Normalized classes of generalized vorticity equations}\label{sec:AdmTrans}

In the course of computing the set of admissible transformations of a class of differential equations, it is often convenient to construct a hierarchy of normalized superclasses for this class~\cite{popo10Ay,popo10By}.
This is why here we also start with the quite general class of differential equations
\begin{gather}\label{eq:class1}
\zeta_t-F(t,x,y,\psi,\psi_x,\psi_y,\zeta,\zeta_x,\zeta_y,\zeta_{xx},\zeta_{xy}, \zeta_{yy})=0, \quad \zeta := \psi_{ii},
\end{gather}
where $(F_{\zeta_x},F_{\zeta_y},F_{\zeta_{xx}},F_{\zeta_{xy}},F_{\zeta_{yy}})\ne(0,0,0,0,0)$, to assure that the generalized vorticity equations of the form~\eqref{eq:vortdirect} belong to this class.
We use notations and agreements from Section~\ref{sec:EquivAlgebra}.
In particular, $z=(z_0,z_1,z_2)=(t,x,y)$, the indices $i$, $j$ and~$k$ again run through $\{1,2\}$, while the indices~$\kappa$, $\lambda$, $\mu$ and $\nu$ range from $0$ to $2$.

Admissible transformations are determined using the direct method in terms of finite transformations.
Namely, we aim to exhaustively describe point transformations of the form
\[
\mathcal T\colon\quad  \tilde z_\mu = Z^\mu(z,\psi), \quad \tilde \psi = \Psi(z,\psi), \quad\mbox{where}\quad J=\frac{\p(Z^0,Z^1,Z^2,\Psi)}{\p(z_0,z_1,z_2,\psi)}\ne0,
\]
which map an equation from class~\eqref{eq:class1} to an equation from the same class.
We express derivatives of the ``old'' dependent variable~$\psi$ with respect to the ``old'' independent variables~$z$
via derivatives of the ``new'' dependent variable~$\tilde\psi$ with respect to the ``new'' independent variables~$\tilde z$.
The latter derivatives will be marked by tilde over~$\psi$.
Thus, the derivative of~$\tilde\psi$ with respect to~$\tilde z_\mu$ is briefly denoted by~$\tilde\psi_\mu$, etc.
Then we substitute the expressions for derivatives into the equation $\zeta_t-F=0$, exclude the new principal derivative~$\tilde\psi_{022}$ using the transformed equation $\tilde\psi_{022}=-\tilde\psi_{011}+\tilde F$, split with respect to parametric variables whenever this is possible and solve the obtained determining equations for~$Z^\mu$ and~$\Psi$ supplemented with the inequality $J\ne0$, considering all arising cases for values of the arbitrary element~$F$ and simultaneously finding the expression for~$\tilde F$ via $F$, $Z^\mu$ and~$\Psi$.

The first order derivatives $\psi_\mu$ are expressed in the following manner:
\[
 \psi_\mu = -\frac{\Psi_\mu-\tilde \psi_\nu Z^\nu_\mu}{\Psi_\psi-\tilde \psi_\nu Z^\nu_\psi} = -\frac{V_\mu}{V_\psi},
\]
where we have introduced the notation $V=V(z,\psi,\tilde z) := \Psi(z,\psi)-\tilde\psi_\nu(\tilde z) Z^\nu(z,\psi)$
which is assumed as a function of the old dependent and independent variables and the new independent variables,
so that $V_\mu=\Psi_\mu-\tilde \psi_\nu Z^\nu_\mu$ and $V_\psi=\Psi_\psi-\tilde \psi_\nu Z^\nu_\psi$.
We will not try to express the old variables via the new variables by inverting the transformation.
This is a conventional trick within the direct method, which essentially simplifies the whole consideration.
In what follows we will also use three more abbreviations similar to $V_\mu$:
\[
U^{\mu\nu}:=Z^\mu_\nu V_\psi - Z^\mu_\psi V_\nu, \quad
W^{\mu\nu} :=U^{\mu i}U^{\nu j} F_{\zeta_{ij}},\quad
P^\mu := U^{\mu 0} - U^{\mu i}F_{\zeta_i}.
\]
Higher order derivatives are expressible in an analogous way. The Laplacian of~$\psi$, e.g., reads
\[
\psi_{ii} = V_\psi^{-3}(U^{\mu i}U^{\nu i}\tilde \psi_{\mu\nu} - V^2_\psi V_{ii} + 2 V_i V_\psi V_{i\psi} - V^2_i V_{\psi\psi}).
\]
For the class~\eqref{eq:class1} considered here, we need the derivatives of the Laplacian up to second order.
The highest derivatives required are of the form
\[
 \psi_{iijk} = V_\psi^{-5}U^{\mu i}U^{\nu i}U^{\kappa j}U^{\lambda k}\tilde \psi_{\mu\nu\kappa\lambda} + \dots,
\]
where the tail contains only derivatives of $\tilde\psi$ up to order three.

Denote by~$G$ the left hand side of the equation obtained by substituting all the expressions for derivatives into~\eqref{eq:class1}.
For the transformation~$\mathcal T$ to be admissible, the condition  $G_{\tilde\psi_{\mu\nu\kappa\lambda}}=0$ has to be satisfied for any tuple of the subscripts $(\mu,\nu,\kappa,\lambda)$ in which at least one of the subscripts equals $0$. Under varying the subscripts, this condition leads to the following system:
\begin{align*}
 &G_{\tilde \psi_{0000}}=0\colon \quad U^{0k}U^{0k}W^{00} = 0, \\
 &G_{\tilde \psi_{000i}}=0\colon \quad U^{0k}U^{0k}W^{0i}+U^{0k}U^{ik}W^{00} = 0,\\
 &G_{\tilde \psi_{00ij}}=0\colon\quad U^{0k}U^{0k}W^{ij}+ 2U^{0k}U^{ik}W^{0j} + 2U^{0k}U^{jk}W^{0i} + U^{ik}U^{jk}W^{00} =0.
\end{align*}
Suppose that $U^{0k}U^{0k}\ne0$. Then the above equations imply that $W^{\mu\nu}:=U^{\mu i}U^{\nu j} F_{\zeta_{ij}}=0$.
If $\rank(U^{\mu i})<2$ then for any~$\mu$ and~$\nu$
\[
U^{\mu1}U^{\nu 2}-U^{\mu2}U^{\nu1}=
\left(\frac{\p(Z^\mu,Z^\nu,\Psi)}{\p(z_1,z_2,\psi)}-\tilde\psi_\kappa\frac{\p(Z^\mu,Z^\nu,Z^\kappa)}{\p(z_1,z_2,\psi)}\right)V_\psi=0
\]
and after splitting with respect to~$\tilde\psi_\lambda$ we obtain that
\[
\frac{\p(Z^\mu,Z^\nu,\Psi)}{\p(z_1,z_2,\psi)}=\frac{\p(Z^\mu,Z^\nu,Z^\kappa)}{\p(z_1,z_2,\psi)}=0
\quad\mbox{or}\quad
Z^\kappa_\psi=\Psi_\psi=0,
\]
but this contradicts the transformation nondegeneracy condition $J\ne0$.
Hence $\rank(U^{\mu i})=2$ and, therefore, the equation $U^{\mu i}U^{\nu j}F_{\zeta_{ij}}=0$ sequentially implies that $U^{\nu j}F_{\zeta_{ij}}=0$ and $F_{\zeta_{ij}}=0$.
Then, the necessary conditions $G_{\tilde \psi_{000}}=0$ and $G_{\tilde \psi_{00i}}=0$ for
admissible transformations are respectively equivalent to the equations
$U^{0k}U^{0k}P^0=0$ and $U^{0k}U^{0k}P^i+2U^{0k}U^{ik}P^0=0$ which jointly gives in view of the condition $U^{0k}U^{0k}\ne0$ that $P^\mu=0$.
Thus, we should have $\det(U^{\mu\nu})=0$. At the same time,
\begin{gather*}
\det(U^{\mu\nu})=V_\psi{}^2\bigl(|Z^0_\nu,Z^1_\nu,Z^2_\nu|V_\psi-|V_\nu,Z^1_\nu,Z^2_\nu|Z^0_\psi-|Z^0_\nu,V_\nu,Z^2_\nu|Z^1_\psi-|Z^0_\nu,Z^1_\nu,V_\nu|Z^2_\psi\bigr)
\\ \phantom{\det(U^{\mu\nu})}{}
=V_\psi{}^2\frac{\p(Z^0,Z^1,Z^2,V)}{\p(z_0,z_1,z_2,\psi)}=V_\psi{}^2J\ne0
\end{gather*}
that leads to a contradiction.
Therefore, the supposition $U^{0k}U^{0k}\ne0$ is not true, i.e., $U^{0k}U^{0k}=0$ and hence $U^{0k}=0$.
Substituting the expressions for $U^{0k}$ and~$V$ into the last equation and splitting with respect to $\tilde\psi_\mu$, we derive the equations
\[
Z^0_kZ^\mu_\psi=Z^0_\psi Z^\mu_k, \quad Z^0_k\Psi_\psi=Z^0_\psi \Psi_k.
\]
The tuples $(Z^\mu_1,\Psi_1)$, $(Z^\mu_2,\Psi_2)$ and $(Z^\mu_\psi,\Psi_\psi)$ are not proportional since $J\ne0$.
This is why we finally obtain the first subset of determining equations $Z^0_k=Z^0_\psi=0$.
It follows from them that $Z^0_0\ne0$ (otherwise $J=0$) and expressions for ``old'' derivatives with respect to only~$x$ and~$y$
contain ``new'' derivatives only of the same type.
In other words, derivatives of~$\tilde\psi$ involving differentiation with respect to~$\tilde t$ appear only in the expressions for $\psi_{0aa}$
and we can simply split with respect to them via collecting their coefficients.

Equating the coefficients of~$\tilde\psi_{012}$ leads, in view of the condition $Z^0_0\ne0$, to the equation $U^{1k}U^{2k}=0$, i.e.,
\begin{equation}\label{EqInProofForAdmTrans1}
\left(Z^1_k\Psi_\psi-Z^1_\psi\Psi_k+(Z^1_\psi Z^2_k-Z^1_kZ^2_\psi)\tilde\psi_2\right)
\left(Z^2_k\Psi_\psi-Z^2_\psi\Psi_k-(Z^1_\psi Z^2_k-Z^1_kZ^2_\psi)\tilde\psi_1\right)=0.
\end{equation}
We split equation~\eqref{EqInProofForAdmTrans1} with respect to~$\tilde\psi_1$ and~$\tilde\psi_2$.
Collecting the coefficients of~$\tilde\psi_1\tilde\psi_2$ gives the equation $(Z^1_\psi Z^2_k-Z^1_kZ^2_\psi)(Z^1_\psi Z^2_k-Z^1_kZ^2_\psi)=0$,
or equivalently $Z^1_\psi Z^2_k-Z^1_kZ^2_\psi=0$.
As $\rank(Z^i_1,Z^i_2,Z^i_\psi)=2$, this implies that $Z^i_\psi=0$ and, therefore, $\Psi_\psi\ne0$.
Consequently, equation~\eqref{EqInProofForAdmTrans1} is reduced to $Z^1_kZ^2_k=0$.

The derivative~$\tilde\psi_{022}$ is assumed principal, $\tilde\psi_{022}=-\tilde\psi_{011}+\tilde F$.
Hence another third order derivative of the above type appropriate for splitting is only $\tilde\psi_{011}$.
The corresponding equation $Z^1_kZ^1_k=Z^2_kZ^2_k:=L$ joint with the equation $Z^1_kZ^2_k=0$ implies
that the functions~$Z^1$ and~$Z^2$ satisfy the Cauchy--Riemann system $Z^1_1=\varepsilon Z^2_2$, $Z^1_2=-\varepsilon Z^2_1$,
where $\varepsilon=\pm1$, and hence $Z^i_{kk}=0$.
Note that $L\ne0$ since $J\ne0$.

Analogously, collecting the coefficients of~$\tilde\psi_{0i}$ and further splitting with respect to~$\tilde\psi_j$ lead to
the equations $Z^i_kZ^j_k\Psi_{\psi\psi}=0$ and $Z^i_k\Psi_{k\psi}=0$. Therefore, $\Psi_{\psi\psi}=0$ and $\Psi_{k\psi}=0$.
Here we take into account the inequalities $L\ne0$ and $\det(Z^i_k)\ne0$.

We do not have more possibilities for splitting.
The derived system of determining equations consists of the equations
\[
Z^0_k=Z^0_\psi=0, \quad Z^i_\psi=0, \quad Z^1_kZ^2_k=0, \quad Z^1_kZ^1_k=Z^2_kZ^2_k, \quad \Psi_{\psi\psi}=\Psi_{k\psi}=0.
\]
The remaining terms determine the transformation rule for the arbitrary element~$F$.
This is why any point transformation satisfying the above determining equations maps every equation from class~\eqref{eq:class1}
to an equation from the same class and, therefore, belongs to the equivalence group~$G^\sim_1$ of class~\eqref{eq:class1}.
In other words, any admissible point transformation of class~\eqref{eq:class1} is induced by a transformation from~$G^\sim_1$,
i.e., class~\eqref{eq:class1} is normalized.
As a result, we have the following theorem.

\begin{theorem}\label{TheoremOnNormalizedSuperclass1}
Class~\eqref{eq:class1} is normalized.
Its equivalence group~$G^\sim_1$ consists of the transformations
\begin{gather*}
\tilde t=T(t), \quad \tilde x=Z^1(t,x,y), \quad \tilde y=Z^2(t,x,y), \quad \tilde\psi=\Upsilon(t)\psi+\Phi(t,x,y), \\
\tilde F=\frac1{T_t}\left(
\frac\Upsilon LF+\Bigl(\frac\Upsilon L\Bigr)_0\zeta+\Bigl(\frac{\Phi_{ii}}L\Bigr)_0
-\frac{Z^i_tZ^i_j}L\left(\frac\Upsilon L\zeta_j+\Bigl(\frac\Upsilon L\Bigr)_j\zeta+\Bigl(\frac{\Phi_{ii}}L\Bigr)_j\right)
\right),
\end{gather*}
where $T$, $Z^i$, $\Upsilon$ and $\Phi$ are arbitrary smooth functions of their arguments,
satisfying the conditions $Z^1_kZ^2_k=0$, $Z^1_kZ^1_k=Z^2_kZ^2_k:=L$, $T_t\Upsilon L\ne0$,
and the subscripts~1 and~2 denote differentiation with respect to~$x$ and~$y$, respectively.
\end{theorem}

The expression for the transformed vorticity is also simple: $\tilde\zeta=L^{-1}(\Upsilon\zeta+\Phi_{ii})$.

\begin{remark*}
The continuous component of unity of the group~$G^\sim_1$ consists of the transformations from~$G^\sim_1$
with $T_t>0$, $\varepsilon=1$ and $\Upsilon>0$.
Therefore, a complete set of independent discrete transformations in~$G^\sim_1$ is exhausted by
the uncoupled changes of the signs of~$t$, $y$ and~$\psi$.
In particular, the value $\varepsilon=-1$ corresponds to alternating the sign of~$y$.
\end{remark*}

Consider the subclass of class~\eqref{eq:class1}, singled out by the constraints $F_\psi=0$, $F_{\psi_x}=-\zeta_y$ and $F_{\psi_y}=\zeta_x$,
i.e., the class consisting of the equations of the form
\begin{equation}\label{eq:class2}
   \zeta_t+\psi_x\zeta_y-\psi_y\zeta_x=H(t,x,y,\zeta,\zeta_x,\zeta_y,\zeta_{xx},\zeta_{xy}, \zeta_{yy}), \quad \zeta := \psi_{ii},
\end{equation}
where $H$ is an arbitrary smooth function of its arguments, which is assumed as an arbitrary element instead of $F=H-\psi_x\zeta_y+\psi_y\zeta_x$.
The class~\eqref{eq:class2} is still a superclass of the class~\eqref{eq:vortdirect}.

\newcommand{\cosbeta}{\mathfrak c}
\newcommand{\sinbeta}{\mathfrak s}

\begin{theorem}\label{TheoremOnNormalizedSuperclass2}
Class~\eqref{eq:class2} is normalized.
The equivalence group~$G^\sim_2$ of this class is formed by the transformations
\begin{equation}\label{EqTransFromGequiv2}
\begin{split}&
\tilde t=\tau, \quad
\tilde x=\lambda(x\cosbeta-y\sinbeta)+\gamma^1, \quad
\varepsilon\tilde y=\lambda(x\sinbeta+y\cosbeta)+\gamma^2, \\&
\tilde\psi=\varepsilon\frac{\lambda}{\tau_t}\left(\lambda\psi+\frac\lambda2\beta_t(x^2+y^2)
-\gamma^1_t(x\sinbeta+y\cosbeta)+\gamma^2_t(x\cosbeta-y\sinbeta)\right)+\delta+\frac\sigma2(x^2+y^2),\\&
\tilde H=\frac\varepsilon{\tau_t{}^2}
\left(H-\frac{\tau_{tt}}{\tau_t}\zeta-\frac{\lambda_t}\lambda(x\zeta_x+y\zeta_y)+2\beta_{tt}-2\frac{\tau_{tt}}{\tau_t}\beta_t
\right)
-\frac{\delta_y+\sigma y}{\tau_t\lambda^2}\zeta_x+\frac{\delta_x+\sigma x}{\tau_t\lambda^2}\zeta_y
\\&\phantom{\tilde F={}}
+\frac2{\tau_t}\left(\frac\sigma{\lambda^2}\right)_t,
\end{split}
\end{equation}
where $\varepsilon=\pm1$, $\cosbeta=\cos\beta$, $\sinbeta=\sin\beta$;
$\tau$, $\lambda$, $\beta$, $\gamma^i$ and $\sigma$ are arbitrary smooth functions of~$t$
satisfying the conditions $\lambda>0$ and $\tau_t\ne0$;
$\delta=\delta(t,x,y)$ runs through the set of solutions of the Laplace equation $\delta_{xx}+\delta_{yy}=0$.
\end{theorem}

\begin{proof}
The class~\eqref{eq:class2} is a subclass of the class~\eqref{eq:class1} and the class~\eqref{eq:class1} is normalized.
Therefore, any admissible transformation of the class~\eqref{eq:class2} is generated by a transformation
from the equivalence group~$G^\sim_1$ of the superclass.
It is only necessary to derive the additional restrictions on transformation parameters caused by narrowing the class.

The group~$G^\sim_1$ is a usual equivalence group~\cite{ovsi82Ay}, i.e., in contrast to different generalizations of equivalence groups~\cite{mele94Ay,popo10Ay},
it consists of point transformations of the joint space of the equation variables and arbitrary elements,
and the components of transformations for the variables do not depend on the arbitrary elements.
Any transformation from~$G^\sim_1$ is additionally projectable to the space of the independent variables
and the space of the single variable~$t$.
This is why it already becomes convenient, in contrast to the proof of Theorem~\ref{TheoremOnNormalizedSuperclass1},
to express the new derivatives via old ones.
Then we substitute the expressions for new derivatives into the transformed equation
$\tilde\zeta_{\tilde t}+\tilde\psi_{\tilde x}\tilde\zeta_{\tilde y}-\tilde\psi_{\tilde y}\tilde\zeta_{\tilde x}=\tilde H,$
exclude the principal derivative~$\psi_{tyy}$ using the equation
\[\psi_{tyy}=-\psi_{txx}-\psi_x\zeta_y+\psi_y\zeta_x+H,\]
split with respect to parametric variables whenever this is possible
and solve the obtained determining equations.
As equations from the class~\eqref{eq:class2} involve derivatives~$\psi_x$ and~$\psi_y$ in an explicitly defined (linear) manner,
we can split with respect to these derivatives, simply collecting their coefficients.
Since these coefficients do not involve the arbitrary element~$H$, we can further split them with respect to other derivatives.
As a result, we obtain the equations
\[
\Upsilon=\varepsilon\frac L{T_t},\quad L_i=0,\quad \Phi_{jji}=0,
\]
where $\varepsilon=\pm1$ and other notations are defined in the proof of Theorem~\ref{TheoremOnNormalizedSuperclass1}.
Therefore, $L$ and $\Phi_{jj}$ are functions of~$t$ only.
As $L>0$, we can introduce the function $\lambda=\sqrt L$ of~$t$.
Acting by the Laplace operator $\p_{jj}$ on the conditions $Z^1_kZ^1_k=\lambda^2$ and $Z^2_kZ^2_k=\lambda^2$ and taking
into account that $Z^i$ are solutions of the Laplace equation, $Z^i_{kk}=0$,
we derive the important differential consequences $Z^i_{jk}=0$, which imply that the functions $Z^i$ are affine in $(x,y)$.
Hence there exists a function $\beta=\beta(t)$ such that
$Z^1_1=\lambda\cosbeta$ and $Z^1_2=-\lambda\sinbeta$, where $\cosbeta=\cos\beta$ and $\sinbeta=\sin\beta$,
and, therefore, $Z^1_1=\varepsilon\lambda\sinbeta$ and $Z^1_2=\varepsilon\lambda\cosbeta$.
We re-denote $T$ by $\tau$ for the sake of notation consistency and represent~$\Phi$ in the following form%
\footnote{There is an ambiguity in representations of~$Z^i$ and~$\Phi$.
For example, the last summand in the representation of~$\Phi$ can be omitted.
The usage of the above complicated representations is motivated by a few reasons:
the consistency with the notation of basis operators of the equivalence algebra~$\mathfrak g^\sim_1$ from Theorem~\ref{TheoremOnEquivAlgebra1},
the simplification of the expression for the transformed arbitrary element~$\tilde H$
and the convenience of studying admissible transformations within subclasses of the class~\eqref{eq:class2}.
}:
\[
\Phi=\delta(t,x,y)+\frac\sigma2(x^2+y^2)+
\varepsilon\frac{\lambda}{\tau_t}\left(\frac\lambda2\beta_t(x^2+y^2)
-\gamma^1_t(x\sinbeta+y\cosbeta)+\gamma^2_t(x\cosbeta-y\sinbeta)\right),
\]
where $\sigma$ is a function of~$t$ and
$\delta=\delta(t,x,y)$ is a solution of the Laplace equation $\delta_{xx}+\delta_{yy}=0$.
Collecting the terms without~$\psi_x$ and~$\psi_y$ gives the transformation for the arbitrary element~$H$.

Similarly to the proof of Theorem~\ref{TheoremOnNormalizedSuperclass1},
any transformation from~$G^\sim_1$ satisfying the above additional constraints maps every equation from the class~\eqref{eq:class2}
to an equation from the same class and, therefore, belongs to the equivalence group~$G^\sim_2$ of the class~\eqref{eq:class2}.
In other words, any admissible point transformation of the class~\eqref{eq:class2} is induced by a transformation from~$G^\sim_2$,
i.e., the class~\eqref{eq:class2} is normalized.
\end{proof}

\begin{remark*}
The transformations from the equivalence group~$G^\sim_2$,
which are associated with the parameter-function~$\delta$ depending only on~$t$, and only such transformations
identically act on the arbitrary element~$H$ and, therefore,
their projections to the space of independent and dependent variables form the kernel (intersection)
of point symmetry groups of the class~\eqref{eq:class2}.
\end{remark*}

\begin{corollary}\label{CorollaryOnNormalizedSuperclass3}
The subclass of the class~\eqref{eq:class2} singled out by the constraint $H_\zeta=0$ is normalized.
Its equivalence group~$G^\sim_3$ consists of the elements of~$G^\sim_2$ with $\tau_{tt}=0$.
\end{corollary}

\begin{proof}
As the vorticity and its derivatives are transformed by elements of~$G^\sim_2$ according to the formulas
\begin{equation}\label{EqTransOfZetaInGequiv2}
\tilde\zeta=\frac\varepsilon{\tau_t}(\zeta+\beta_t)+2\frac\sigma{\lambda^2},\quad
\tilde\zeta_i=\frac{\varepsilon Z^i_j}{\tau_t\lambda^2}\zeta_j,
\end{equation}
it follows from~\eqref{EqTransFromGequiv2} under the constraints $H_\zeta=0$ and $\tilde H_{\tilde\zeta}=0$ that $\tau_{tt}=0$.
The rest of the proof is similar to the end of the proof of Theorem~\ref{TheoremOnNormalizedSuperclass2}.
\noprint{
Again, any transformation of the form~\eqref{EqTransFromGequiv2} with $\tau_{tt}=0$ maps every equation from
the subclass of the class~\eqref{eq:class2} singled out by the constraint $H_\zeta=0$
to an equation from the same subclass and, therefore, belongs to the equivalence group~$G^\sim_3$ of the subclass.
In other words, any admissible point transformation of this subclass is induced by a transformation from~$G^\sim_2$,
i.e., the subclass is normalized.
}
\end{proof}

\begin{corollary}\label{CorollaryOnNormalizedSuperclass4}
The subclass of the class~\eqref{eq:class2} singled out by the constraints $H_i=0$ is normalized.
Its equivalence group~$G^\sim_4$ consists of the elements of~$G^\sim_2$ with $\lambda_t=0$, $\sigma=0$ and $\delta_{ij}=0$.
\end{corollary}

\begin{proof}
As any admissible transformation of the class~\eqref{eq:class2} has the form~\eqref{EqTransFromGequiv2} and, therefore,
the vorticity and its derivatives are transformed according to~\eqref{EqTransOfZetaInGequiv2},
the system $\tilde H_{\tilde x}=0$, $\tilde H_{\tilde y}=0$ is equivalent to the system $\tilde H_x=0$, $\tilde H_y=0$.
After differentiating the last equation in~\eqref{EqTransFromGequiv2} with respect to~$x$ and~$y$
and splitting with respect to~$\zeta_x$ and~$\zeta_y$, we derive all the above additional constraints on transformation parameters.
The rest of the proof is similar to the end of the proof of Theorem~\ref{TheoremOnNormalizedSuperclass2}.
\end{proof}

\begin{corollary}\label{CorollaryOnNormalizedSuperclass5}
The subclass of the class~\eqref{eq:class2} singled out by the constraints $H_\zeta=0$ and $H_i=0$ is normalized.
Its equivalence group~$G^\sim_5$ consists of the elements of~$G^\sim_2$ with  $\tau_{tt}=0$, $\lambda_t=0$, $\sigma=0$ and $\delta_{ij}=0$.
\end{corollary}

\begin{proof}
The subclass under consideration is normalized as it is the intersection of the normalized subclasses from
Corollaries~\ref{CorollaryOnNormalizedSuperclass3} and~\ref{CorollaryOnNormalizedSuperclass4}.
Therefore, we also have $G^\sim_5=G^\sim_3\cap G^\sim_4$.
\end{proof}

\begin{remark*}
For the subclass from Corollary~\ref{CorollaryOnNormalizedSuperclass5},
the kernel of point symmetry groups is essentially extended in comparison with the whole class~\eqref{eq:class2}.
It is formed by the projections of elements of the equivalence group~$G^\sim_2$,
associated with the parameter-functions~$\gamma^1$ and~$\gamma^2$ and the parameter-function~$\delta$ depending only on~$t$,
to the space of independent and dependent variables, cf.\ Section~\ref{SectionOnParameterizationViaDirectGroupClassification}.
\end{remark*}

A further narrowing is given by the condition that the arbitrary element~$H$ with $H_\zeta=0$ is a total divergence with respect to the space variables,
i.e., $H=\mathrm D_i f^i$ for some differential functions $f^i=f^i(t,x,y,\zeta_x,\zeta_y)$.
The corresponding subclass rewritten in the terms of~$f^i$ coincides with the class~\eqref{eq:vortdirect}
and is singled out from the class~\eqref{eq:class2} by the constraints $H_\zeta=0$ and ${\sf E}H=0$,
where ${\sf E}=\p_\zeta-\mathrm D_i\p_{\zeta_i}+\sum_{i\leqslant j}\mathrm D_i\mathrm D_j\p_{\zeta_{ij}}+\dots$ is the associated Euler operator.
In this Euler operator, the role of independent and dependent variables is played by $(x,y)$ and $\zeta$, respectively,
and the variable $t$ is assumed as a parameter.
The vorticity~$\zeta$ can be considered in~$\sf E$ as the dependent variable instead of~$\psi$ since
the arbitrary element~$H$ depends only on combinations of derivatives of~$\psi$ being derivatives of~$\zeta$.

\begin{remark*}
It is obvious that the arbitrary element~$H$ satisfies the constraints $H_\zeta=0$ and ${\sf E}H=0$ if it is represented in the
form $H=\mathrm D_i f^i$ for some differential functions $f^i=f^i(t,x,y,\zeta_x,\zeta_y)$.
The converse claim should be proved.
Thus, the constraint ${\sf E}H=0$ implies the representation $H=\mathrm D_i f^i$
for some differential functions $f^i(t,x,y,\zeta,\zeta_x,\zeta_y)$, which may depend on~$\zeta$.
Substituting this representation into the constraint $H_\zeta=0$ and splitting the resulting equations with respect to
the second derivatives of~$\zeta$, we obtain the following system of PDEs for the functions~$f^i$:
$f^i_{\zeta i}+f^i_{\zeta\zeta}\zeta_i=0$, $f^1_{\zeta\zeta_1}=0$, $f^2_{\zeta\zeta_2}=0$, $f^1_{\zeta\zeta_2}+f^2_{\zeta\zeta_1}=0$.
Its general solution has the form $f^1= D_2\Psi+\tilde f^1$ and $f^2=-D_1\Psi+\tilde f^2$
for some smooth functions $\Psi=\Psi(t,x,y,\zeta)$ and $\tilde f^i=\tilde f^i(t,x,y,\zeta_x,\zeta_y)$.
The first summands in the expressions for~$f^i$ can be neglected due to the gauge equivalence
in the set of arbitrary elements $(f^1,f^2)$.
As a result, we construct the necessary representation for the arbitrary element~$H$.%
\end{remark*}

\begin{corollary}\label{CorollaryOnNormalizedSuperclass6}
The class~\eqref{eq:vortdirect} is normalized.
The equivalence group~$G^\sim_6$ of this class represented in terms of the arbitrary element~$H$
consists of the elements of~$G^\sim_2$ with $\tau_{tt}=0$ and $\lambda_t=0$.
The arbitrary elements $f^i$ are transformed in the following way:
\begin{gather}\label{EqEquivTransForFi}
\begin{split}
&\tilde f^1=\varepsilon\lambda\frac{f^1\cosbeta-f^2\sinbeta}{\tau_t^2}
+\biggl(\frac{\delta}{\tau_t\lambda}+\frac{\sigma}{2\tau_t\lambda}(x^2+y^2)-\frac{\varepsilon\chi}{\lambda^2}\biggr)(\zeta_x\sinbeta+\zeta_y\cosbeta)
\\
&\phantom{\tilde f^1=}
{}+(\varepsilon\lambda^2\beta_{tt}+\tau_t\sigma_t)\frac{x\cosbeta-y\sinbeta}{\tau_t{}^2\lambda}
-\varepsilon\frac{\rho_x\sinbeta+\rho_y\cosbeta}{\lambda^2},
\\
&\tilde f^2=\lambda\frac{f^1\sinbeta+f^2\cosbeta}{\tau_t{}^2}
-\varepsilon\biggl(\frac{\delta}{\tau_t\lambda}+\frac{\sigma}{2\tau_t\lambda}(x^2+y^2)-\frac{\varepsilon\chi}{\lambda^2}\biggr)(\zeta_x\cosbeta-\zeta_y\sinbeta)
\\
&\phantom{\tilde f^2=}
{}+\varepsilon(\varepsilon\lambda^2\beta_{tt}+\tau_t\sigma_t)\frac{x\sinbeta+y\cosbeta}{\tau_t{}^2\lambda}
+\frac{\rho_x\cosbeta-\rho_y\sinbeta}{\lambda^2},
\end{split}
\end{gather}
where $\chi=\chi(t)$ and $\rho=\rho(t,x,y)$ are arbitrary functions of their arguments.
\end{corollary}

\begin{proof}
The class~\eqref{eq:vortdirect} is contained in the normalized subclass of the class~\eqref{eq:class2} singled out by the constraint $H_\zeta=0$.
Therefore, any admissible transformation of the class~\eqref{eq:vortdirect} is generated by an element of~$G^\sim_2$ with $\tau_{tt}=0$,
and the corresponding transformations of the space variables are affine with respect to these variables, $Z^i_{jk}=0$.
Then $\tilde{\mathrm D}_j\tilde f^j=\mathrm D_i(\lambda^{-2}Z^j_i\tilde f^j)$,
i.e., the differential function~$\smash{\tilde H}$ is a total divergence with respect to the new space variables
if and only if it is a total divergence with respect to the old space variables.
Applying the Euler operator~$\sf E$ to the last equality in~\eqref{EqTransFromGequiv2} under the conditions
$H_\zeta=0$, $\tilde H_{\tilde\zeta}=0$, ${\sf E}H=0$, $\tilde {\sf E}\tilde H=0$ and $\tau_{tt}=0$,
we derive the additional constraint $\lambda_t=0$.
The remaining part of the proof of normalization of the class~\eqref{eq:vortdirect} and its equivalence group is analogous to
the end of the proof of Theorem~\ref{TheoremOnNormalizedSuperclass2}.

In order to construct the transformations of the arbitrary elements~$f^i$, we represent the right hand side of the last equality
in~\eqref{EqTransFromGequiv2} as a total divergence: $\tilde H=\mathrm D_ih^i$, where
\begin{gather*}
h^1=\frac{\varepsilon}{\tau_t^2}(f^1+\beta_{tt}x)
+\frac{\sigma_t}{\tau_t\lambda^2}x+\biggl(\delta+\frac\sigma2(x^2+y^2)\biggr)\frac{\zeta_y}{\tau_t\lambda^2},
\\
h^2=\frac{\varepsilon}{\tau_t^2}(f^2+\beta_{tt}y)
+\frac{\sigma_t}{\tau_t\lambda^2}y-\biggl(\delta+\frac\sigma2(x^2+y^2)\biggr)\frac{\zeta_x}{\tau_t\lambda^2},
\end{gather*}
As $\tilde H=\tilde{\mathrm D}_j\tilde f^j=\mathrm D_ih^i=\tilde{\mathrm D}_jZ^j_ih^i$,
the pair of the differential functions $\tilde f^j-Z^j_ih^i$ is a null divergence, $\tilde{\mathrm D}_i(\tilde f^j-Z^j_ih^i)=0$.
In view of Theorem 4.24 from~\cite{olve86Ay} 
there exists a differential function~$Q$ depending on $t$, $x$, $y$ and derivatives of~$\zeta$ such that
$\tilde f^1-Z^1_ih^i=-\tilde{\mathrm D}_2Q$ and $\tilde f^2-Z^2_ih^i=\tilde{\mathrm D}_2Q$.
As $\tilde{\mathrm D}_iQ$ and, therefore, $\mathrm D_iQ$ should be functions of $t$, $x$, $y$, $\zeta_x$ and $\zeta_y$,
the function $Q$ is represented in the form $Q=\chi(t)\zeta+\rho(t,x,y)$ for some smooth functions $\chi=\chi(t)$ and $\rho=\rho(t,x,y)$.
\looseness=-1
\end{proof}

\begin{remark*}
The equivalence transformations associated with the parameter-functions~$\chi$ and~$\rho$
are identical with respect to both the independent and dependent variables,
i.e., they transform only arbitrary elements with no effect on the corresponding equation
and, therefore, are \emph{trivial} \cite[p.~53]{lisl92Ay} 
or \emph{gauge} \cite[Section 2.5]{popo10Ay} equivalence transformations.
These transformations arise due to the special representation of the arbitrary element~$H$ as a total divergence
and form a normal subgroup of the entire equivalence group considered in terms of the arbitrary elements $f^1$ and~$f^2$,
called the \emph{gauge equivalence group} of the class~\eqref{eq:vortdirect}.
\end{remark*}

\begin{remark*}
The continuous component of unity of the group~$G^\sim_6$ is singled out from~$G^\sim_6$ by the conditions $\tau_t>0$ and $\varepsilon=1$.
Therefore, a complete set of independent discrete transformations in~$G^\sim_6$ is exhausted by
alternating signs either in the tuple $(t,\psi)$ or in the tuple $(y,\psi,f^1)$.
\end{remark*}

Consider the subclass of the class~\eqref{eq:vortdirect}, singled out by the further auxiliary equation~$f^i_j = 0$, i.e.,
the class of equations
\begin{equation}\label{eq:vortdirectSpaceInd}
\zeta_t + \{\psi,\zeta\} = \mathrm D_i f^i(t,\zeta_x,\zeta_y), \quad \zeta := \psi_{ii},
\end{equation}
with the arbitrary elements $f^i=f^i(t,\zeta_x,\zeta_y)$.

\begin{remark*}
Rewritten in the terms of~$H$, the class~\eqref{eq:vortdirectSpaceInd} is a well-defined subclass of~\eqref{eq:class2}.
It is singled out from the class~\eqref{eq:class2} by the constraints
${\sf E}H=0$, $H_\zeta=0$, $H_i=0$ and $\zeta_{ij}H_{\zeta_{ij}}=H$.
Indeed, the representation $H=\mathrm D_i f^i(t,\zeta_x,\zeta_y)$ obviously implies that
the arbitrary element~$H$ does not depend on~$x$, $y$ and~$\zeta$, is annulated by the Euler operator~$\sf E$ and
is a (homogenous) linear function in the totality of the derivatives~$\zeta_{ij}$.
Hence all the above constraints are necessary.
Conversely, the constraint ${\sf E}H=0$ implies that the arbitrary element~$H$ is affine in the totality of~$\zeta_{ij}$
and, therefore, in view of the constraint $\zeta_{ij}H_{\zeta_{ij}}=H$ it is a (homogenous) linear function in these derivatives of~$\zeta$.
As a result, we have the representation $H=h^{ij}\zeta_{ij}$,
where the coefficients $h^{ij}$, $h^{12}=h^{21}$, depend solely on~$t$, $\zeta_x$ and~$\zeta_y$ since $H_\zeta=0$ and $H_i=0$.
Then the constraint ${\sf E}H=0$ is equivalent to the single equation \[2h^{12}_{\zeta_1\zeta_2}=h^{11}_{\zeta_2\zeta_2}+h^{22}_{\zeta_1\zeta_1}\]
whose general solutions is represented in the form
$h^{11}=f^1_{\zeta_1}$, $h^{12}=f^1_{\zeta_2}+f^2_{\zeta_1}$ and $h^{22}=f^2_{\zeta_2}$ for some differential functions
$f^i=f^i(t,\zeta_x,\zeta_y)$. This finally gives the necessary representation for~$H$.
\end{remark*}

\begin{remark*}
In view of the previous remark, the subclass of the class~\eqref{eq:class2},
singled out by the constraints ${\sf E}H=0$, $H_\zeta=0$ and $H_i=0$ is a proper
superclass for the class~\eqref{eq:vortdirectSpaceInd} rewritten in the terms of~$H$.
This superclass of~\eqref{eq:vortdirectSpaceInd} is normalized since it is the intersection of
the normalized class from Corollary~\ref{CorollaryOnNormalizedSuperclass5} and the normalized class~\eqref{eq:vortdirect}.
Its equivalence group coincides with the group~$G^\sim_5$
described in Corollary~\ref{CorollaryOnNormalizedSuperclass5}.
\end{remark*}

In a way analogous to the above proofs,
the normalization of the superclass and formulas~\eqref{EqTransFromGequiv2} and~\eqref{EqEquivTransForFi}
imply the following assertion.

\begin{corollary}\label{CorollaryOnNormalizedSuperclass7}
The class~\eqref{eq:vortdirectSpaceInd} is normalized.
The equivalence group~$G^\sim_7$ of this class represented in terms of the arbitrary element~$H$
consists of the elements of~$G^\sim_2$ with $\tau_{tt}=0$, $\lambda_t=0$, $\beta_{tt}=0$, $\sigma=0$ and $\delta_i=0$.
The arbitrary elements $f^i$ are transformed according to~\eqref{EqEquivTransForFi},
where additionally $\rho_{ij}=0$.
\end{corollary}

\begin{remark*}
The above consideration of normalized classes is intended for the description of invariant parameterizations
of the forms~\eqref{eq:vortdirect} and~\eqref{eq:vortdirectSpaceInd}.
The hierarchy of normalized classes constructed is, in some sense, minimal and optimal for this purpose.
It can be easily extended with related normalized classes.
For instance, the subclass singled out from the class~\eqref{eq:class2} by the constraints ${\sf E}H=0$ is normalized.
Other hierarchies of normalized classes, which are related to the vorticity equation~\eqref{eq:vort}
and different from the hierarchy presented, can be constructed.
\end{remark*}

\begin{remark*}
In fact, all subclasses of generalized vorticity equations studied in this section are strongly normalized,
cf.~\cite{popo10Ay}.
\end{remark*}

\subsection{Parameterization via direct group classification}\label{SectionOnParameterizationViaDirectGroupClassification}

As proved in Section~\ref{sec:AdmTrans} (see Corollary~\ref{CorollaryOnNormalizedSuperclass7}), the class~\eqref{eq:vortdirectSpaceInd} is normalized.
Hence the complete group classification for this class can be obtained within the algebraic method.
Another way to justify the sufficiency of the algebraic method is to check the weak normalization of the class~\eqref{eq:vortdirectSpaceInd}
in infinitesimal sense, i.e., the condition that the linear span the maximal Lie invariance algebras of equations from the class~\eqref{eq:vortdirectSpaceInd}
is contained in the projection of its equivalence algebra~$\mathfrak g^\sim_2$ (cf.\ Section~\ref{sec:EquivAlgebra}) to the space of independent and dependent variables. A vector field~$Q$ in the space of the variables $(t,x,y,\psi)$ has the form $Q=\xi^\mu\p_\mu+\eta\p_\psi$,
where the coefficients $\xi^\mu$ and $\eta$ smoothly depend on $(t,x,y,\psi)$.
For $Q$ to be a Lie symmetry operator of an equation from the class~\eqref{eq:vortdirectSpaceInd}, its coefficients should satisfy the following system of determining equations that do not involve the arbitrary elements $(f^1,f^2)$:
\begin{gather*}
\xi^\mu_\psi=0,\quad \xi^0_i = 0, \quad \xi^0_{tt}=0,  \quad \xi^i_{jk}=0, \quad \xi^1_{1t}=0, \quad \xi^1_1=\xi^2_2, \quad \xi^1_2=-\xi^2_1, \\
\eta_{\psi\psi} = 0, \quad \eta_{\psi t}=0, \quad \eta_{\psi 1}=\xi^1_t, \quad \eta_{\psi 2}=-\xi^2_t, \quad \eta_\psi-2\xi^1_1 + \xi^0_t = 0,
\end{gather*}
The integration of the above system immediately implies that $Q\in\pr\mathfrak g^\sim_2$.

The equivalence algebra~$\mathfrak g^\sim_2$ can be represented as a semidirect sum
$\mathfrak g^\sim_2=\tilde{\mathfrak i}\mathbin{\mbox{$\lefteqn{\hspace{.67ex}\rule{.4pt}{1.2ex}}{\ni}$}}\tilde{\mathfrak a}$,
where
$\tilde{\mathfrak i}=\langle \tilde\XX(\gamma^1),\,\tilde\YY(\gamma^2),\,\tilde{\mathcal Z}(\chi)\rangle$
and
$\tilde{\mathfrak a}=\langle \tilde \DDD_1, \tilde \DDD_2, \p_t, \tilde\JJ^1, \tilde\JJ^t,\,
\tilde{\mathcal K}(\delta),\,\tilde{\mathcal G}(\rho^1x+\rho^2y)\rangle$
are an ideal and a subalgebra of~$\mathfrak g^\sim_2$, respectively.
Here $\gamma^1$, $\gamma^2$, $\rho^1$, $\rho^2$, $\delta$ and~$\chi$ run through the set of smooth functions of the variable~$t$
and we use the notation $\tilde\JJ^1=\tilde\JJ(1)$, $\tilde\JJ^t=\tilde\JJ(t)$ and
$\tilde{\mathcal K}(\delta)=\tilde{\mathcal H}(\delta)-\tilde{\mathcal Z}(\delta)$.
The intersection (\emph{kernel}) of the maximal Lie invariance algebras of equations from class~\eqref{eq:vortdirectSpaceInd} is
\[
\mathfrak g^\cap_2=\langle\XX(\gamma^1),\,\YY(\gamma^2),\,\ZZ(\chi)\rangle=\pr\,\tilde{\mathfrak i}.
\]
In other words, the complete infinite dimensional part $\pr\,\tilde{\mathfrak i}$ of the projection of the equivalence algebra~$\mathfrak g^\sim_2$
to the space of variables $(t,x,y,\psi)$ is already a Lie invariance algebra for any equation from the class~\eqref{eq:vortdirectSpaceInd}.
Therefore, any Lie symmetry extension is only feasible via (finite-dimensional) subalgebras of the five-dimensional solvable algebra
\[
\mathfrak a=\langle \DDD_1,\, \p_t,\, \DDD_2,\, \JJ,\, \JJ^t\rangle=\pr\tilde{\mathfrak a}.
\]
In other words, for any values of the arbitrary elements $f^i=f^i(t,\zeta_x,\zeta_y)$
the maximal Lie invariance algebra $\mathfrak g^{\rm max}_f$ of the corresponding equation~$\mathcal L_f$ from the class~\eqref{eq:vortdirectSpaceInd}
is represented in the form $\mathfrak g^{\rm max}_f=\mathfrak g^{\rm ext}_f\mathbin{\mbox{$\lefteqn{\hspace{.67ex}\rule{.4pt}{1.2ex}}{\in}$}}\mathfrak g^\cap_2$,
where $\mathfrak g^{\rm ext}_f$ is a subalgebra of $\mathfrak a$.
A nonzero linear combination of the operators~$\JJ$ and~$\JJ^t$ is a Lie symmetry operator of the equation~$\mathcal L_f$ if and only if
this equation is invariant with respect to the algebra $\langle\JJ,\JJ^t\rangle$.
Therefore, for any extension within the class~\eqref{eq:vortdirectSpaceInd} we have that
either $\mathfrak g^{\rm ext}_f\cap\langle\JJ,\JJ^t\rangle=\{0\}$ or $\mathfrak g^{\rm ext}_f\supset\langle\JJ,\JJ^t\rangle$, i.e.,
\begin{equation}\label{eq:ConditionForExtensionsWithJAndJt}
\dim(\mathfrak g^{\rm ext}_f\cap\langle\JJ,\JJ^t\rangle)\in\{0,2\}.
\end{equation}

Moreover, as $\pr\mathfrak g^\sim_2=\mathfrak g_0$, the maximal Lie invariance algebra of the inviscid barotropic vorticity equation~\eqref{eq:vort}, the normalization of class~\eqref{eq:vortdirectSpaceInd} means that only subalgebras of~$\mathfrak g_0$ can be used to construct spatially independent parameterization schemes within the class~\eqref{eq:vortdirectSpaceInd}. That is, for such parameterizations, the approach from \cite{ober00Ay} based on inverse group classification is quite natural and gives the same exhaustive result as direct group classification.
Due to the normalization, the complete realization of preliminary group classification of equations from the class~\eqref{eq:vortdirectSpaceInd} is also equivalent to its direct group classification which can be carried out for this class with the algebraic method.

Note that the class~\eqref{eq:vortdirectSpaceInd} possesses the nontrivial gauge equivalence algebra
\[\mathfrak g^{\rm gauge}=\langle\tilde{\mathcal K}(\delta),\,\tilde{\mathcal G}(\rho^1x+\rho^2y)\rangle, \]
cf.\ the second remark after Theorem~\ref{TheoremOnEquivAlgebra1}.
As we have $\pr\mathfrak g^{\rm gauge}=\{0\}$, the projections of operators from $\mathfrak g^{\rm gauge}$
obviously do not appear in $\mathfrak g^{\rm ext}_f$ for any value of $f$. At the same time, they are essential for finding all possible parameterizations that admit symmetry extensions.

Therefore two equivalent ways for the further use of the algebraic method in this problem depending on subalgebras of what algebra will be classified.

As a first impression, the optimal way is to construct a complete list of inequivalent subalgebras of the Lie algebra~$\mathfrak a$
and then substitute basis operators of each obtained subalgebra to the infinitesimal invariance criterion
in order to derive the associated system of equations for $f^i$ that should be integrated.
The algebra~$\mathfrak a$ is  finite dimensional and has the structure of a direct sum, $\mathfrak a=\langle\DDD_1,\,\p_t,\,\JJ,\,\JJ^t\rangle\oplus\langle\DDD_2\rangle$. The first summand is the four-dimensional Lie algebra $\mathfrak g_{4.8}^{-1}$ in accordance with Mubarakzyanov's classification of low-dimensional Lie algebras~\cite{muba63Ay} whose nilradical is isomorphic to the Weyl (Bianchi II) algebra $\mathfrak g_{3.1}$.
The classification of inequivalent subalgebra up to the equivalence relation generated by the adjoint action of the corresponding Lie group on~$\mathfrak a$
is a quite simple problem.
Moreover, the set of subalgebras to be used is reduced after taking into account the condition~\eqref{eq:ConditionForExtensionsWithJAndJt}.
At the same time, the derived systems for $f^i$ consist of second order partial differential equations and
have to be integrated up to $G^\sim_7$-equivalence.

This is why another way is optimal.
It is based on the fact that $\smash{\mathfrak g^{\rm ext}_f}$ coincides with a subalgebra~$\mathfrak b$ of~$\mathfrak a$
if and only if there exists a subalgebra~$\tilde{\mathfrak b}$ of~$\tilde{\mathfrak a}$
such that $\pr\tilde{\mathfrak b}=\mathfrak b$ and
the arbitrary elements $f^i$ satisfy the equations
\begin{equation}\label{eq:InvariantSurfaceCondForFi}
\xi^0f^i_t+\theta^jf^i_{\zeta_j}=\varphi^i
\end{equation}
for any operator~$\smash{\tilde{\mathcal Q}}$ from~$\smash{\tilde{\mathfrak b}}$, where
$\xi^0$, $\theta^j$ and~$\varphi^i$ are coefficients of $\p_t$, $\p_{\zeta_j}$ and~$\p_{f^i}$ in~$\smash{\tilde{\mathcal Q}}$, respectively.
In fact, the system~\eqref{eq:InvariantSurfaceCondForFi} is the invariant surface condition for the operator~$\tilde{\mathcal Q}$
and the functions~$f^i$ depending only on~$t$ and~$\zeta^j$.
This system is not compatible for any operator from~$\tilde{\mathfrak a}$
of the form $\tilde{\mathcal Q}=\tilde{\mathcal K}(\delta)+\tilde{\mathcal G}(\rho^1x+\rho^2y)$,
where at least one of the parameter-functions $\delta$, $\rho^1$ or~$\rho^2$ does not vanish.
In other words, each operator from~$\tilde{\mathfrak b}$ should have a nonzero part belonging to
$\langle \tilde \DDD_1, \tilde \DDD_2, \p_t, \tilde\JJ^1, \tilde\JJ^t\rangle$
and hence $\dim\pr\tilde{\mathfrak b}=\dim\tilde{\mathfrak b}\leqslant5$.
Taking into account also the condition~\eqref{eq:ConditionForExtensionsWithJAndJt},
we obtain the following algorithm for classification of possible Lie symmetry extensions within the class~\eqref{eq:vortdirectSpaceInd}:
\begin{enumerate}\itemsep=-.5ex
\item
We classify $G^\sim_7$-inequivalent subalgebras of the algebra~$\tilde{\mathfrak a}$ each of which satisfies the conditions
$\dim\pr\tilde{\mathfrak b}=\dim\tilde{\mathfrak b}$ and $\dim(\tilde{\mathfrak b}\cap\langle\JJ,\JJ^t\rangle)\in\{0,2\}$.
Adjoint actions corresponding to operators from~$\tilde{\mathfrak i}$ can be neglected.
\item
We fix a subalgebra~$\tilde{\mathfrak b}$ from the list constructed in the first step.
This algebra is necessarily finite dimensional, $\dim\tilde{\mathfrak b}\leqslant5$.
We solve the system consisting of equations of the form~\eqref{eq:InvariantSurfaceCondForFi},
where the operator~$\tilde{\mathcal Q}$ runs through a basis of~$\tilde{\mathfrak b}$.
For every solution of this system we have $\mathfrak g^{\rm ext}_f=\pr\tilde{\mathfrak b}$.
\item
Varying~$\tilde{\mathfrak b}$, we get the required list of values of
the arbitrary elements $(f^1,f^2)$ and the corresponding Lie symmetry extensions.
\end{enumerate}

In order to realize the first step of the algorithm, we list the nonidentical adjoint actions related to basis elements of $\tilde{\mathfrak a}$:
\begin{align*}
 &\Ad(e^{\ve\p_t})\DDD_1 = \DDD_1 - \ve\p_t,&                                         &\Ad(e^{\ve\DDD_1})\p_t = e^\ve\p_t, \\
 &\Ad(e^{\ve\JJ^t})\DDD_1 = \DDD_1 + \ve\JJ^t,&                                       &\Ad(e^{\ve\DDD_1})\JJ^t = e^{-\ve}\JJ^t, \\
 &\Ad(e^{\ve\mathcal K(\delta)})\DDD_1 = \DDD_1 + \ve\mathcal K(t\delta_t+\delta), &  &\Ad(e^{\ve\DDD_1})\mathcal K(\delta) = \mathcal K(e^{-\ve}\delta(e^{-\ve}t)), \\
 &\Ad(e^{\ve\mathcal G(\rho)})\DDD_1 = \DDD_1 + \ve\mathcal G(t\rho_t+2\rho), &       &\Ad(e^{\ve\DDD_1})\mathcal G(\rho) = \mathcal G(e^{-2\ve}\rho(e^{-\ve}t,x,y)), \\
 &\Ad(e^{\ve\mathcal K(\delta)})\p_t = \p_t + \ve\mathcal K(\delta_t), &              &\Ad(e^{\ve\p_t})\JJ^t = \JJ^t - \ve\JJ^1, \\
 &\Ad(e^{\ve\mathcal G(\rho)})\p_t = \p_t + \ve\mathcal G(\rho_t),&                   &\Ad(e^{\ve\p_t})\mathcal K(\delta) = \mathcal K(\delta(t-\ve)), \\
 &\Ad(e^{\ve\JJ^t})\p_t = \p_t + \ve\JJ, &                                            &\Ad(e^{\ve\p_t})\mathcal G(\rho) = \mathcal G(\rho(t-\ve,x,y)), \\
 &\Ad(e^{\ve\mathcal K(\delta)})\DDD_2 = \DDD_2 + \ve\mathcal K(2\delta), &           &\Ad(e^{\ve\DDD_2})\mathcal K(\delta) = \mathcal K(e^{2\ve}\delta(t)), \\
 &\Ad(e^{\ve\mathcal G(\rho)})\DDD_2 = \DDD_2 + \ve\mathcal G(2\rho), &               &\Ad(e^{\ve\DDD_2})\mathcal G(\rho) = \mathcal G(e^{-\ve}\rho(t,e^{-\ve}x,e^{-\ve}y)) ,\\
 &\Ad(e^{\ve\mathcal G(\rho)})\JJ^1 = \JJ^1 + \ve\mathcal G(\rho^2x-\rho^1y), &       &\Ad(e^{\ve\JJ^1})\mathcal G(\rho) = \mathcal G(\hat\rho^\ve), \\
 &\Ad(e^{\ve\mathcal G(\rho)})\JJ^t = \JJ^t + \ve\mathcal G(t\rho^2x-t\rho^1y), &     &\Ad(e^{\ve\JJ^t})\mathcal G(\rho) = \mathcal G(\check\rho^\ve),
\end{align*}
where we omit tildes in the notation of operators and also omit arguments of parameter-functions if these arguments are not changed under the corresponding adjoint action,
$\rho=\rho^1x+\rho^2y$,
$\hat\rho^\ve=(\rho^1x+\rho^2y)\cos\ve+(\rho^1y-\rho^2x)\sin\ve$,
$\check\rho^\ve=(\rho^1x+\rho^2y)\cos\ve t+(\rho^1y-\rho^2x)\sin\ve t$,

Based upon these adjoint actions, we derive the following list of
$G^\sim_7$-inequivalent subalgebras of~$\tilde{\mathfrak a}$ satisfying the above restrictions
(we again omit tildes in the notation of operators):

\medskip

\noindent\textit{one-dimensional subalgebras:}
\begin{equation*}
  \langle\DDD_1+b\DDD_2+a\JJ^1\rangle, \quad
  \langle\p_t+c\DDD_2+\hat c\JJ^t\rangle, \quad
  \langle\DDD_2+\JJ^t\rangle,\quad
  \langle\DDD_2+a\JJ^1\rangle;
\end{equation*}

 \noindent\textit{two-dimensional subalgebras:}
\begin{align*}
  &\langle\DDD_1+b\DDD_2+a\JJ+\mathcal K (c)+\mathcal G(\tilde c x),\p_t\rangle, \quad \langle \DDD_1+a\JJ^1,\DDD_2+\hat a\JJ^1\rangle, \\
  &\langle\p_t+c\JJ^t, \DDD_2+\hat a\JJ^1\rangle, \quad \langle\JJ^1+\mathcal K(\delta^1(t)),\JJ^t+\mathcal K(\delta^2(t))\rangle;
\end{align*}

\noindent\textit{three-dimensional subalgebras:}
\begin{align*}
  &\langle\DDD_1+a\JJ^1, \p_t, \DDD_2+\hat a\JJ^1\rangle,\quad
  \langle\DDD_1+b\DDD_2,\JJ^1+\mathcal K(c|t|^{2b-1}),\JJ^t+\mathcal K(\hat c|t|^{2b})\rangle, \\
  &\langle\p_t+\tilde c\DDD_2,\JJ^1+\mathcal K(ce^{2\tilde ct}),\JJ^t+\mathcal K((ct+\hat c)e^{2\tilde ct})\rangle, \quad
  \langle\DDD_2,\JJ^1,\JJ^t\rangle;
\end{align*}

\noindent\textit{four-dimensional subalgebras:}
\begin{align*}
  &\langle\DDD_1+b\DDD_2+\mathcal K(\nu_2), \p_t, \JJ^1+\mathcal K(\nu_1), \JJ^t+\mathcal K(\nu_1t+\nu_0) \rangle, \ (2b-1)\nu_1=0, \ b\nu_0 = 0, \\
  &\langle \DDD_1, \DDD_2, \JJ^1, \JJ^t\rangle, \quad
  \langle \p_t, \DDD_2, \JJ^1, \JJ^t\rangle;
  \end{align*}

\noindent\textit{five-dimensional subalgebra:}
\begin{align*}
    \langle \DDD_1, \p_t, \DDD_2, \JJ^1, \JJ^t\rangle.
  \end{align*}
In the above subalgebras, due to adjoint actions we can put the following restrictions on the algebra parameters: $a\ge0$, $c,\tilde c\in\{0,1\}$, $\hat a \ge 0$ if $a=0$ (resp.\ $c=0$), $\hat c\in\{0,1\}$ if $c=0$; additionally, in the first two-dimensional subalgebra we can set $(1+2b)c=0$ and $((1+b)^2+a^2)\tilde c=0$; in the first four-dimensional subalgebra one non-zero parameter among $\nu_0$, $\nu_1$, $\nu_2$ can be set to $1$. In the last two-dimensional subalgebra, the parameters $\delta^1$ and $\delta^2$ are arbitrary smooth functions of $t$. The subalgebras with parameter tuples $(\delta^1, \delta^2)$ and $(\tilde \delta^1, \tilde \delta^2)$ are equivalent if and only if there exist constants $\ve_0$, $\ve_1$ and $\ve_2$ such that $\tilde \delta^1=e^{\ve_2-\ve_1}\delta^1(e^{-\ve_1}t+\ve_0)$ and $\tilde \delta^2=e^{\ve_2}\delta^2(e^{-\ve_1}t+\ve_0)$.

Concerning the realization of the second step of the algorithm, we note that the system corresponding to the last two-dimensional subalgebra is compatible if and only if $\delta^2(t)=t\delta^1(t)$. We re-denote $\delta^1$ by $\delta$. As the general solution of the system is parameterized by functions of two arguments, we put the associated two-dimensional symmetry extension into Table~\ref{tab:OneDimensionalSymmetryExtensionsRestrictedCase}, where the other extensions are one-dimensional. A similar remark is true for the three last three-dimensional subalgebras, which is why we list them in Table~\ref{tab:TwoDimensionalSymmetryExtensionsRestrictedCase} containing symmetry extensions parameterized by functions of a single argument.

The system associated with the first two-dimensional subalgebra is compatible if and only if $(a,b)\ne(0,-1)$. The solution of the system is split into three cases, (i) $b\ne-1,1/2$, (ii) $b=1/2$ and (iii) $b=-1$ and $a\ne0$. We will use the notation $\mu = c/(2b-1)$ for $b\ne1/2$ and $\mu = 2c/3$ in case of $b=1/2$.

For the second and third three-dimensional subalgebras, the corresponding systems are compatible if and only if $c=\hat c$ and $\hat c=0$, respectively.

For the reason of compatibility, in the first four-dimensional subalgebra we have $\nu_0=0$ and $b\ne-1$. Due to the condition $(2b-1)\nu_1=0$, the solution of the corresponding system should be split into the two cases $b\ne1/2$ and $b=1/2$. For simplicity of the representation of the results in Table~\ref{tab:ThreeDimensionalSymmetryExtensionsRestrictedCase} we introduce the notation $\mu = \nu_2/(2b-1)$ if $b\ne1/2$ and $\tilde \nu_2 = 2\nu_2/3$ for $b=1/2$.

\smallskip

\begin{table}[!htp]%
\centering
\caption{Symmetry extensions parameterized by functions of two arguments\label{tab:OneDimensionalSymmetryExtensionsRestrictedCase}}\vspace{1ex}
\begin{tabular}{llll}
\hline
\hfil $\mathfrak g^{\rm ext}_f$ &\hfil Arguments of $I_1$, $I_2$ &&\hfil $f^1$, $f^2$ \\
\hline
\multirow{2}{*}{$\langle\DDD_1+b\DDD_2+a\JJ\rangle$} & $|t|^{b+1}(\zeta_x\cos\tau +\zeta_y\sin\tau ),$ & \multirow{2}{*}{$\tau:=a\ln|t|$} & $|t|^{b-2}(I_1\cos\tau-I_2\sin\tau), $  \\
& $|t|^{b+1}(\zeta_y\cos\tau -\zeta_x\sin\tau ),$ && $|t|^{b-2}(I_1\sin\tau+I_1\cos\tau)$
\\
\multirow{2}{*}{$\langle\p_t+c\DDD_2+\hat c\JJ^t\rangle$} & $e^{ct}(\zeta_x\cos\tau+\zeta_y\sin\tau),$ &\multirow{2}{*}{$\tau:=\dfrac{\hat c}2 t^2$} & $e^{ct}(I_1\cos\tau-I_2\sin\tau),$ \\
 & $e^{ct}(\zeta_y\cos\tau-\zeta_x\sin\tau),$ && $e^{ct}(I_1\sin\tau+I_1\cos\tau)$
\\
$\langle\DDD_2+\JJ^t\rangle$ & $t,\ Re^{\Phi/t}$ && $P_1, P_2$
\\
$\langle\DDD_2+a\JJ\rangle$ & $t,\ R^ae^{\Phi}$ && $P_1, P_2$
\\
\multirow{2}{*}{$\langle\JJ,\JJ^t\rangle$}  & \multirow{2}{*}{$t,\ R$} && $\zeta_xI^1-\zeta_y I^2+\delta(t)\zeta_y\Phi$, \\
& && $\zeta_yI^1+\zeta_x I^2-\delta(t)\zeta_x\Phi$
\\
\hline
\end{tabular}
\end{table}

\begin{table}[!htp]%
\centering
\caption{Symmetry extensions parameterized by functions of a single argument\label{tab:TwoDimensionalSymmetryExtensionsRestrictedCase}}\vspace{1ex}
\begin{tabular}{lll}
\hline
\hfil $\mathfrak g^{\rm ext}_f$ &\hfil Argument of $I_1$, $I_2$ &\hfil $f^1$, $f^2$ \\
\hline
$\langle\DDD_1+b\DDD_2+a\JJ,\p_t\rangle$, $b\ne-1,\tfrac12$ & $R^ae^{(1+b)\Phi}$ &  $R^{\alpha_1}P_1-\mu\zeta_y$, $R^{\alpha_1}P_2+\mu\zeta_x$ \\
$\langle\DDD_1+\tfrac12\DDD_2+a\JJ,\p_t\rangle$  & $R^ae^{3\Phi/2}$ &  $R^2P_1-\mu\zeta_y\ln R$, $R^2P_2+\mu\zeta_x\ln R$ \\
$\langle\DDD_1-\DDD_2+a\JJ,\p_t\rangle$, $a\ne0$ & $R$ &  $e^{\alpha_2\Phi}P_1-\mu\zeta_y$, $e^{\alpha_2\Phi}P_2+\mu\zeta_x$ \\
$\langle\DDD_1+a\JJ,\DDD_2+\hat a\JJ\rangle $ & $|t|^{\hat a - a}R^{\hat a}e^\Phi $ & $t^{-3}P_1$, $t^{-3}P_2$\\
$\langle\p_t+c\JJ^t,\DDD_2+\hat a\JJ\rangle $ & $R^{\hat a}e^{\Phi-ct^2/2}$ & $P_1$, $P_2$  \\
\multirow{2}{*}{$\langle\DDD_1+b\DDD_2,\JJ,\JJ^t\rangle$} & \multirow{2}{*}{$|t|^{b+1}R $} & $|t|^{2b-1}(\zeta_xI^1-\zeta_yI^2+c\zeta_y\Phi)$, \\
& & $|t|^{2b-1}(\zeta_yI^1+\zeta_xI^2-c\zeta_x\Phi)$\\
\multirow{2}{*}{$\langle\p_t+\tilde c\DDD_2, \JJ,\JJ^t\rangle$} & \multirow{2}{*}{$e^{\tilde c t}R$} & $e^{2\tilde ct}(\zeta_x I^1-\zeta_y I^2+c\zeta_y \Phi)$, \\
& & $e^{2\tilde ct}(\zeta_y I^1+\zeta_x I^2-c\zeta_x \Phi)$\\
$\langle\DDD_2,\JJ,\JJ^t\rangle$ & $t$ & $P_1$, $P_2$ \\
\hline
\end{tabular}
\end{table}

\begin{table}[!htp]%
\centering
\caption{Symmetry extensions parameterized by constants\label{tab:ThreeDimensionalSymmetryExtensionsRestrictedCase}}\vspace{1ex}
\begin{tabular}{ll}
\hline
\hfil $\mathfrak g^{\rm ext}_f$ & \hfil $f^1$, $f^2$ \\
\hline
$\langle\DDD_1+a\JJ^1, \p_t, \DDD_2+\hat a\JJ^1\rangle$, $\hat a\ne a$ & $R^{\alpha_3\hat a}e^{\alpha_3\Phi}P_1$, $R^{\alpha_3\hat a}e^{\alpha_3\Phi}P_2$ \\
$\langle\DDD_1+b\DDD_2,\p_t,\JJ,\JJ^t\rangle$, $b\ne-1,\tfrac12$ & $R^{\alpha_1}P_1-\mu\zeta_y$, $R^{\alpha_1}P_2+\mu\zeta_x$ \\
$\langle\DDD_1+\tfrac12\DDD_2,\p_t,\JJ,\JJ^t\rangle$ & $R^2P_1+(\tilde\nu_2\ln R+\nu_1\Phi)\zeta_y$, $R^2P_2-(\tilde\nu_2\ln R+\nu_1\Phi)\zeta_x$ \\
$\langle\DDD_1,\DDD_2,\JJ,\JJ^t\rangle$ & $t^{-3}P_1$, $t^{-3}P_2$ \\
$\langle\p_t,\DDD_2,\JJ,\JJ^t \rangle$ &  $P_1$, $P_2$\\
\hline
\end{tabular}
\end{table}

In Tables \ref{tab:OneDimensionalSymmetryExtensionsRestrictedCase}--\ref{tab:ThreeDimensionalSymmetryExtensionsRestrictedCase},
$I_1$ and $I_2$ are arbitrary functions of two indicated arguments, arbitrary functions of one indicated argument or arbitrary constants, respectively,
\[
R = \sqrt{\zeta_x^2+\zeta_y^2},\quad
\Phi = \arctan\frac{\zeta_y}{\zeta_x},\quad
P_1 = \frac{\zeta_xI_1-\zeta_yI_2}{\zeta_x^2+\zeta_y^2},\quad
P_2=\frac{\zeta_yI_1+\zeta_xI_2}{\zeta_x^2+\zeta_y^2}.
\]
Moreover, $\alpha_1=3/(b+1)$ (for $b\ne-1$), $\alpha_2=3/a$ (for $b=-1$ and $a\ne0$) and $\alpha_3=3/(\hat a-a)$ (for $\hat a\ne a$).
In Table~\ref{tab:OneDimensionalSymmetryExtensionsRestrictedCase}, $\delta$ is an arbitrary function of~$t$.

Up to gauge equivalence, the single parameterization admitting five-dimensional symmetry extension within the class~\eqref{eq:vortdirectSpaceInd} is the trivial parameterization, $f^1=f^2=0$, in which we neglect the eddy vorticity flux. This shows the limits of applicability of the method proposed in~\cite{ober97Ay}, cf.\ Section~\ref{sec:ParameterizationInverseGroupClassification}.

\subsection{Parameterization via preliminary group classification}\label{sec:ParameterizationViaPreliminaryGroupClassification}

The technique of preliminary group classification is based on classifications of extensions of the kernel Lie invariance algebra by operators obtained via projection of elements of the corresponding equivalence algebra to the space of independent and dependent variables~\cite{ibra91Ay}. It is illustrated here with the class~\eqref{eq:vortdirect} whose equivalence algebra~$\mathfrak g^\sim_1$ is calculated in Section~\ref{sec:EquivAlgebra}.

The kernel Lie invariance algebra~$\mathfrak g^\cap_1$ of the class~\eqref{eq:vortdirect}
(i.e., the intersection of the maximal Lie invariance algebras of equations from the class) is $\langle\ZZ(\chi)\rangle$.
Denote by $\tilde{\mathfrak g}^\cap_1$ the ideal of $\mathfrak g^\sim_1$ corresponding to~$\mathfrak g^\cap_1$,
$\tilde{\mathfrak g}^\cap_1=\langle\tilde\ZZ(\chi)\rangle$.
In view of the classification of one-dimensional subalgebras of the equivalence algebra in Appendix~\ref{sec:appendix1} (list~\eqref{eq:OneDimensionalSubalgebrasEquivalenceAlgebra}) and since for preliminary group classification we are only concerned with extensions of the complement of~$\tilde{\mathfrak g}^\cap_1$ in~$\mathfrak g^\sim_1$, we essentially have to consider the inequivalent subalgebras
\begin{align*}
\begin{split}
 &\langle\tilde\DDD_1+a\tilde\DDD_2\rangle,\quad \langle\p_t+b\tilde\DDD_2\rangle, \quad
  \langle\tilde\DDD_2+\tilde\JJ(\beta)+\tilde{\mathcal R}(\sigma)\rangle, \quad
  \langle\tilde\JJ(\beta)+\tilde{\mathcal R}(\sigma)\rangle, \\
 &\langle\tilde\XX(\gamma^1)+\tilde{\mathcal R}(\sigma)\rangle,\quad
  \langle\tilde{\mathcal R}(\sigma)+\tilde{\mathcal H}(\delta)+\tilde{\mathcal G}(\rho)\rangle.
\end{split}
\end{align*}
Here $a\in\mathbb R$, $b\in\{-1,0,1\}$, $\beta=\beta(t)$, $\sigma=\sigma(t)$, $\gamma^1=\gamma^1(t)$ and $\rho=\rho(t,x,y)$ are smooth functions of their arguments and $\delta=\delta(t,x,y)$ is a solution of the Laplace equation, $\delta_{xx}+\delta_{yy}=0$. All parameters are arbitrary but fixed for a particular subalgebra.
For each of the subalgebras, the corresponding arbitrary elements $f^i$ satisfy the equations
\begin{equation}\label{eq:InvariantSurfaceCondForFiGen}
\xi^\mu f^i_\mu+\theta^jf^i_{\zeta_j}=\varphi^i
\end{equation}
where $\xi^\mu$, $\theta^j$ and~$\varphi^i$ respectively are coefficients of $\p_\mu$, $\p_{\zeta_j}$ and~$\p_{f^i}$ in the basis element of the subalgebra.
It now remains to present the parameterization schemes constructed, which can be found in the Table~\ref{tab:OneDimensionalSymmetryExtensionsGeneralCase}.

\begin{table}[!htp]
\centering
\caption{One-dimensional symmetry algebra extensions for the case $f^i=f^i(t,x,y,\zeta_x,\zeta_y)$\label{tab:OneDimensionalSymmetryExtensionsGeneralCase}}\vspace{1ex}
\begin{tabular}{lll}
\hline
\hfil Extension &\hfil Arguments of $I_1$, $I_2$ &\hfil $f^1$, $f^2$ \\
\hline
$\langle\DDD_1+a\DDD_2\rangle$ & $|t|^{-a}x$, $|t|^{-a}y$, $tx\zeta_x$, $ty\zeta_y$ & $t^{-2}xI_1$, $t^{-2}yI_2$
\\[1ex]
$\langle\p_t+a\DDD_2\rangle$ & $e^{-at}x$, $e^{-at}y$, $x\zeta_x$, $y\zeta_y$ & $xI_1$, $yI_2$
\\[1ex]
\raisebox{-2.5ex}[0ex][0ex]{$\langle\DDD_2+\JJ(\beta)+\mathcal R(\sigma)\rangle$}
&  \raisebox{-1ex}[0ex][0ex]{\parbox[t]{21ex}{$t$, $\varphi-\beta\ln r$,\\ $x\zeta_x+y\zeta_y$, $y\zeta_x-x\zeta_y$}} &
$xI^1-yI^2+\dfrac{\sigma}{2}r^2\zeta_y\ln r + (\beta_{tt}+\sigma_t)x\ln r,$\\[1ex]
& & $yI^1+xI^2-\dfrac{\sigma}{2}r^2\zeta_x\ln r + (\beta_{tt}+\sigma_t)y\ln r$
\\[2ex]
\raisebox{-1ex}[0ex][0ex]{\parbox[t]{18ex}{$\langle\JJ(\beta)+\mathcal R(\sigma)\rangle$,\\$\beta\ne0$}}
& \raisebox{-1ex}[0ex][0ex]{\parbox[t]{21ex}{$t$, $r$,\\ $x\zeta_x+y\zeta_y$, $y\zeta_x-x\zeta_y$}}
& $xI^1-yI^2+\dfrac{\sigma}{2\beta}r^2 \zeta_y\varphi + \dfrac{\beta_{tt}+\sigma_t}{\beta}x\varphi$, \\[1.5ex]
&
& $yI^1+xI^2-\dfrac{\sigma}{2\beta}r^2 \zeta_x\varphi + \dfrac{\beta_{tt}+\sigma_t}{\beta}y\varphi$
\\[2ex]
\raisebox{-1ex}[0ex][0ex]{\parbox[t]{18ex}{$\langle\XX(\gamma^1)+\mathcal R(\sigma)\rangle$,\\$\gamma^1\ne0$}}
& \multirow{2}{*}{$t$, $y$, $\zeta_x$, $\zeta_y$}
&  $I_1 + \dfrac{\sigma_t}{\gamma^1}\dfrac{x^2}{2}+\dfrac{\sigma\zeta_y}{6\gamma^1}(x^3+3xy^2)$, \\[1.5ex]
&& $I_2 + \dfrac{\sigma_t}{\gamma^1}xy-\dfrac{\sigma\zeta_x}{6\gamma^1}(x^3+3xy^2)$\\[1.5ex]
\hline
\end{tabular}
\end{table}
In this table, $I_1$ and $I_2$ are arbitrary functions of four indicated arguments, $r=\sqrt{x^2+y^2}$, $\varphi = \arctan{y/x}$ and
$\mathcal R(\sigma)=\pr\tilde{\mathcal R}(\sigma)=\frac12\sigma r^2\p_\psi$.

In the last class of subalgebras no ansatz can be constructed due to the special form of functions~$f^i$. Namely, as the variable $\psi$ is not included in the list of arguments of $f^i$, any nonzero operator of the form $\tilde{\mathcal R}(\sigma)+\tilde{\mathcal H}(\delta)+\tilde{\mathcal G}(\rho)$ gives an incompatible system of the form~\eqref{eq:InvariantSurfaceCondForFiGen} and hence its projection does not belong to Lie invariance algebras of equations from the class~\eqref{eq:vortdirect}.

Note that some of the extensions presented are not maximal even for the general values of the invariant functions~$I^1$ and~$I^2$.
In particular, if an equation from the class~\eqref{eq:vortdirect} possesses a Lie symmetry operator of the form $\XX(\gamma^1)$
with a fixed function~$\gamma^1$, it possesses all the operators of this form.

As the class~\eqref{eq:vortdirect} is normalized (see Corollary~\ref{CorollaryOnNormalizedSuperclass6}),
its complete group classification also can be obtained by the algebraic method.
For this it is enough to classify only special subalgebras of the equivalence algebra~$\mathfrak g^\sim_1$,
cf.\ a similar classification in Section~\ref{SectionOnParameterizationViaDirectGroupClassification} which also is used here.
The restrictions for appropriate subalgebras are mentioned above under the classification of (at least) one-dimensional Lie symmetry extensions.
Now we precisely formulate them:
\begin{itemize}\itemsep=0ex
\item
The projection of a subalgebra~$\mathfrak s$ of~$\mathfrak g^\sim_1$ to the space of variables~$(t,x,y,\psi)$ is a Lie invariance algebra of an equation from  the class~\eqref{eq:vortdirect} if and only if the corresponding system of equations of the form~\eqref{eq:InvariantSurfaceCondForFiGen} for the arbitrary elements $f^1$ and $f^2$ is compatible.
\item
Only subalgebras of~$\mathfrak g^\sim_1$ should be classified whose projections to the space of variables $(t,x,y,\psi)$ are the maximal Lie invariance algebra of certain equations from  the class~\eqref{eq:vortdirect}.
\end{itemize}
As a result, the classification is split into several cases.
For each of the cases we have a common part of Lie symmetry extensions, which may be infinite dimensional.
All additional extensions are finite dimensional and can be classified with reasonable efforts.
We briefly describe only the main cases arising under the classification.
The complete classification will be presented elsewhere.

Let $\mathfrak s$ be an appropriate subalgebra of~$\mathfrak g^\sim_1$.
As remarked above, any appropriate subalgebra does not contain nonzero operators of the form $\tilde{\mathcal R}(\sigma)+\tilde{\mathcal H}(\delta)+\tilde{\mathcal G}(\rho)$ and includes $\tilde{\mathfrak g}^\cap_1=\langle\tilde\ZZ(\chi)\rangle$ as a proper ideal.
Denote by $\mathfrak j$ the subspace of~$\mathfrak g^\sim_1$ spanned by the operators $\tilde\XX(\gamma^1)$, $\tilde\YY(\gamma^2)$, $\tilde{\mathcal R}(\sigma)$, $\tilde{\mathcal H}(\delta)$ and $\tilde{\mathcal G}(\rho)$, where the parameters runs through the corresponding sets of functions, cf.\ Theorem~\ref{TheoremOnEquivAlgebra1}.
Then denote by $r_0$ the rank of the set of tuples of functional parameters $(\gamma^1,\gamma^2)$
appearing in operators from $\mathfrak s\cap\mathfrak j$.
It is obvious that $r_0\in\{0,1,2\}$.
We consider each of the possible values of~$r_0$ separately.

If $r_0=0$, any nonzero operator from~the complement of $\tilde{\mathfrak g}^\cap_1$ in $\mathfrak s$ has a nonzero projection to the subalgebra $\langle\tilde\DDD_1,\tilde\DDD_2,\p_t,\tilde\JJ(\beta)\rangle$, where the parameter-function~$\beta$ runs through the set of smooth functions of~$t$.
Suppose that the operators $Q^a=\tilde\JJ(\beta^a)+T^a$ with fixed linearly independent functions~$\beta^a$ and tails
$T^a=\tilde\XX(\gamma^{a1})+\tilde\YY(\gamma^{a2})+\tilde{\mathcal R}(\sigma^a)+\tilde{\mathcal H}(\delta^a)+\tilde{\mathcal G}(\rho^a)\in\mathfrak j,$
$a=1,\dots,n$, where $n\geqslant2$, belong to~$\mathfrak s$.
Up to $G^\sim_6$-equivalence we can assume that $T^1=\tilde{\mathcal R}(\sigma^1)$, i.e., $\gamma^{11}=0$, $\gamma^{12}=0$, $\delta^1=0$ and $\rho^1=0$.
As $r_0=0$, we have that the commutator of any pair of operators $Q$'s should be a linear combinations of certain $Q$'s and $\tilde\ZZ(\chi)$.
This condition taken for $Q^1$ and the other $Q$'s implies that $\gamma^{a1}=\gamma^{a2}=0$.
Denote by $\mathcal E^a$ the equation of the form~\eqref{eq:InvariantSurfaceCondForFiGen}, associated with the operator~$Q^a$.
For each $a\ne1$, we subtract the equation $\mathcal E^1$ multiplied by~$\beta^a$ from the equation $\mathcal E^a$ multiplied by~$\beta^1$.
This results in the equation that does not involves $f^i$ and, therefore, is an identity.
Splitting it with respect to~$\zeta_x$ and~$\zeta_y$, we obtain the system
\begin{gather*}
(\beta^1\sigma^a-\beta^a\sigma^1)(x^2+y^2)+2\beta^1\delta^a=0,\\
(\beta^1\beta^a_{tt}-\beta^a\beta^1_{tt}+\beta^1\sigma^a_t-\beta^a\sigma^1_t)x+\beta^1\rho^a_y=0,\\
(\beta^1\beta^a_{tt}-\beta^a\beta^1_{tt}+\beta^1\sigma^a_t-\beta^a\sigma^1_t)y-\beta^1\rho^a_x=0,
\end{gather*}
Taking into account that $\delta^a_{xx}+\delta^a_{yy}=0$ and cross differentiating the two last equations of the system,
we then derive that $\beta^1\sigma^a-\beta^a\sigma^1=0$, $\delta^a=0$, $\rho^a_x=0$, $\rho^a_y=0$ and
\begin{equation}\label{EqJbetaSystemForParameters2}
\beta^1\beta^a_{tt}-\beta^a\beta^1_{tt}+\beta^1\sigma^a_t-\beta^a\sigma^1_t=0.
\end{equation}
Since the parameter-functions $\rho^a$ are defined up to summands being arbitrary smooth functions of~$t$,
we can assume, in view of the equations $\rho^a_x=0$ and $\rho^a_y=0$, that $\rho^a=0$.
The equations $\beta^1\sigma^a-\beta^a\sigma^1=0$ mean that the tuples of~$\beta$'s and~$\sigma$'s are proportional for each~$t$,
i.e., there exists a smooth function~$\alpha=\alpha(t)$ such that $(\sigma^1,\dots,\sigma^n)=\alpha(\beta^1,\dots,\beta^n)$.
We combine the last condition with equations~\eqref{EqJbetaSystemForParameters2} and solve the resulting equations
\[
(\beta^1\beta^a_t-\beta^a\beta^1_t)_t+\alpha(\beta^1\beta^a_t-\beta^a\beta^1_t)=0
\]
with respect to $\beta^a$. The solutions are
$
\beta^a=c_{1a}\beta^1\int(\beta^1)^{-2}\tilde\alpha\,dt+c_{2a}\beta^1,
$
where $\tilde\alpha=e^{-\int\alpha\,dt}$ and $c_{1a}$ and $c_{2a}$ are arbitrary constants.
Therefore, the number~$n$ of linearly independent functions~$\beta^a$ cannot be greater than~2.
Summing up the above consideration, we conclude that
basis elements of $\mathfrak s$ belonging to the complement of $\tilde{\mathfrak g}^\cap_1$
can be assumed to have the following form:
\[
S^b+\tilde\JJ(\hat\beta^b)+\hat T^b, \quad b=1,\dots,m, \quad \tilde\JJ(\beta^a)+\tilde{\mathcal R}(\sigma^a), \quad a=1,\dots,n,
\]
where $\langle S^b,\,b=1,\dots,m\rangle$ is an $m$-dimensional subalgebra of $\langle\tilde\DDD_1,\tilde\DDD_2,\p_t\rangle$
and hence $0\leqslant m\leqslant3$, $\hat T^b\in\mathfrak j$, $0\leqslant n\leqslant2$, the functions~$\beta^a$ are linearly independent
and $\sigma^a=-(\ln|\beta^1\beta^2_t-\beta^2\beta^1_t|)_t\beta^a$ if $n=2$.
The total dimension of extension in this case equals $m+n$ and is not greater than~5.

The condition $r_0=1$ implies that the subalgebra~$\mathfrak s$ contains no operators of the form $\tilde\JJ(\beta)+T$, where $\beta\ne0$ and $T\in\mathfrak j$.
Suppose that operators
$
T^s=\tilde\XX(\gamma^{s1})+\tilde\YY(\gamma^{s2})+\tilde{\mathcal R}(\sigma^s)+\tilde{\mathcal H}(\delta^s)+\tilde{\mathcal G}(\rho^s)
$
from~$\mathfrak j$, $s=1,\dots,p$, where $p\geqslant2$ and $(\gamma^{s1},\gamma^{s2})$ are linearly independent pairs of functions, belong to~$\mathfrak s$.
Up to $G^\sim_6$-equivalence we can assume that $\gamma^{12}=0$, $\delta^1=0$ and $\rho^1=0$.
As $r_0=1$, this also means that $\gamma^{\varsigma2}=0$, $\varsigma=2,\dots,p$, and the parameter-functions $\gamma^{s1}=0$, $s=1,\dots,p$, are linearly independent.
Analogously to the previous case,
denote by $\mathcal E^s$ the equation of the form~\eqref{eq:InvariantSurfaceCondForFiGen}, associated with the operator~$T^s$.
For each $s\ne1$, we subtract the equation $\mathcal E^1$ multiplied by~$\gamma^{s1}$ from the equation $\mathcal E^s$ multiplied by~$\gamma^{11}$.
This results in the equation that does not involves $f^i$ and, therefore, is an identity.
Making the same manipulations with the identity as those in the previous case, we obtain $\delta^s=0$, $\rho^s=0$,
$\gamma^{11}\sigma^s_t=\gamma^{s1}\sigma^1_t$, $\gamma^{11}\sigma^s=\gamma^{s1}\sigma^1$ and, therefore, $\gamma^{11}_t\sigma^s=\gamma^{s1}_t\sigma^1$.
In view of the linear independence of~$\gamma^{s1}$ and~$\gamma^{11}$,
the last two conditions form a well-determined homogenous system of linear algebraic equations with respect to~$\sigma^1$ and~$\sigma^s$ and hence imply that
$\sigma^s=\sigma^1=0$.
At the same time, if an equation from the class~\eqref{eq:vortdirect} possesses a Lie symmetry operator $\XX(\gamma^1)$
with a fixed function~$\gamma^1$, it possesses all the operators of this form.
This means that there are only two $G^\sim_6$-inequivalent possibility for $\mathfrak s\cap\mathfrak j$ in this case,
namely, $\mathfrak s\cap\mathfrak j$ is either spanned by a single operator $\tilde\XX(\gamma^{01})+\tilde{\mathcal R}(\sigma^0)$,
where~$\gamma^{01}$ and~$\sigma^0$ are fixed smooth nonvanishing functions of~$t$,
or equal to the entire set of operators of the form $\tilde\XX(\gamma^1)$,
where~$\gamma^1$ runs through the set of smooth functions of~$t$.
Additional extensions are realized only by tuple of operators of the form
$S^b+\tilde\JJ(\hat\beta^b)+\hat T^b$, $b=1,\dots,m$,
where $\hat T^b\in\mathfrak j$, $\langle S^b,\,b=1,\dots,m\rangle$ is an $m$-dimensional subalgebra of $\langle\tilde\DDD_1,\tilde\DDD_2,\p_t\rangle$
and hence $0\leqslant m\leqslant3$.

Let $r_0=2$. We use notations of the previous case and assume summation for the repeated index~$i$.
Suppose that operators~$T^s$, $s=1,\dots,p$, where $p\geqslant3$ and $(\gamma^{s1},\gamma^{s2})$ are linearly independent pairs of functions, belong to~$\mathfrak s$.
In view of the condition $r_0=2$, up to permutation of the operators~$T^s$ we can assume without loss of generality that $\gamma^{11}\gamma^{22}-\gamma^{12}\gamma^{21}\ne0.$
Then for each $s>2$ there exist smooth functions $\alpha^{si}$ of~$t$, $i=1,2$, such that
$(\gamma^{s1},\gamma^{s2})=\alpha^{si}(\gamma^{i1},\gamma^{i2})$.
Subtracting the equation $\mathcal E^i$ multiplied by~$\alpha^{si}$ from the equation $\mathcal E^s$,
we derive the equation which should identically satisfied and, therefore, implies after certain manipulations that
$\delta^s=\alpha^{si}\delta^i$, $\rho^s=\varrho^i$, $\sigma^s=\alpha^{si}\sigma^i$, $\sigma^s_t=\alpha^{si}\sigma^i_t$
and hence $\alpha^{si}_t\sigma^i=0$.
We should separately consider two subcases depending on either vanishing or nonvanishing $\sigma^i\sigma^i$.

If $\sigma^i\sigma^i\ne0$ then $\mathfrak s\cap\mathfrak j$ coincides with the set of operators of the general form
\[
\tilde\XX(\alpha^i\gamma^{i1})+\tilde\YY(\alpha^i\gamma^{i2})+
\tilde{\mathcal R}(\alpha^i\sigma^i)+\tilde{\mathcal H}(\alpha^i\delta^i)+\tilde{\mathcal G}(\alpha^i\rho^i),
\]
where~$(\alpha^1,\alpha^2)$ runs through the set of pairs of smooth functions of~$t$ satisfying the condition $\alpha^i_t\sigma^i=0$.
In view of commutation relations between $\tilde\JJ(\beta)$ and operators from~$\mathfrak j$,
no operator of the form $\tilde\JJ(\beta)+T$, where $\beta\ne0$ and $T\in\mathfrak j$, belongs to~$\mathfrak s$.
Additional extensions are realized only by tuple of operators of the form
$S^b+\tilde\JJ(\hat\beta^b)+\hat T^b$, $b=1,\dots,m$,
where $\hat T^b\in\mathfrak j$, $\langle S^b,\,b=1,\dots,m\rangle$ is an $m$-dimensional subalgebra of $\langle\tilde\DDD_1,\tilde\DDD_2,\p_t\rangle$
and hence $0\leqslant m\leqslant3$.

Suppose that $\sigma^1=\sigma^2=0$.
The condition $[T^1,T^2]\in\mathfrak s$ implies that
\begin{gather*}
\gamma^{11}\delta^2_x+\gamma^{12}\delta^2_y=\gamma^{21}\delta^1_x+\gamma^{22}\delta^1_y,\quad
\gamma^{11}\rho^2_x+\gamma^{12}\rho^2_y=\gamma^{21}\rho^1_x+\gamma^{22}\rho^1_y.
\end{gather*}
Therefore, using the push-forwards of transformations from~$G^\sim_6$, we can set $\delta^i=0$, $\rho^i=0$.
In other words, we can assume that the subalgebra~$\mathfrak s$ contains the operators
$T^i=\tilde\XX(\gamma^{i1})+\tilde\YY(\gamma^{i2})$, where $\gamma^{11}\gamma^{22}-\gamma^{12}\gamma^{21}\ne0$.
The system of equations of the form~\eqref{eq:InvariantSurfaceCondForFiGen}, associated with these operators,
is equivalent to the system $f^i_x=f^i_y=0$, $i=1,2$,
which singles out the subclass~\eqref{eq:vortdirectSpaceInd} from the class~\eqref{eq:vortdirect}.
The complete group classification of this subclass has been carried out in Section~\ref{SectionOnParameterizationViaDirectGroupClassification}.

\section{Conclusion}\label{sec:conclusion}

In this paper we have addressed the question of symmetry-preserving parameterization schemes. It was demonstrated that the problem of finding invariant parameterization schemes can be treated as a group classification problem. In particular, the interpretation of parameterizations as particular elements of classes of differential equations renders it possible to use well-established methods of symmetry analysis for the design of general classes of closure schemes with prescribed symmetry properties. For parameterizations to admit selected subgroups of the maximal Lie invariance group of the unaveraged differential equation, they should be expressed in terms of related differential invariants. The general outline of this approach is depicted in Figure~\ref{fig:schemeInverse}. Differential invariants can be computed either using infinitesimal methods or the method of moving frames, cf.\ Section~\ref{sec:ParameterizationInverseGroupClassification}.

It should be stressed that the selection of subgroups with respect to which a parameterization scheme should be invariant can be naturally justified when considering boundary-value problems. It is usually the case that explicitly taking into account particular initial and/or boundary conditions strongly decreases the number of admitted symmetries, see e.g.~\cite{bihl11By,bihl12By,blum89Ay} for further discussions and particular examples related to geophysical fluid dynamics. For selected subgroups not to be trivial, one can consider a class of similar boundary-value problems instead of a fixed problem and selected those symmetries that are extended to equivalence transformations of this class of boundary-value problems.
Hence, symmetry-subgroup admitting parameterization schemes can be especially useful when a parameterization scheme is constructed for particular boundary-value problems.

\begin{figure}[!htpb]
\centering
\hspace*{5mm}\includegraphics[scale=0.8]{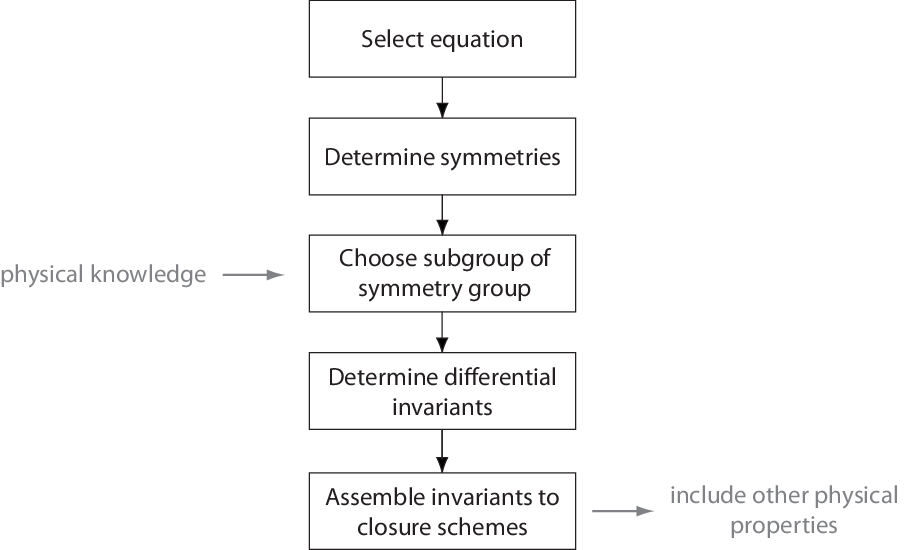}
\caption{Schematic overview of the construction of invariant parameterization schemes based on methods of inverse group classification.
\label{fig:schemeInverse}}
\end{figure}

For parameterization ansatzes with prescribed functional dependence on the resolved quantities and no prescribed symmetry group, the direct group classification problem should be solved. In the case where the given class of differential equations is normalized (which can be checked by the computation of the set of admissible transformations), it is possible and convenient to carry out the classification using the algebraic method~\cite{popo10Ay}. In the case where the class fails to be normalized (or in the case where it is impossible to compute the set of admissible transformations), an exhaustive investigation of parameterizations might be possible due to applying compatibility analysis of the corresponding determining equations or by combining the algebraic and compatibility methods. For more involved classes of differential equations at least symmetry extensions induced by subalgebras of the equivalence algebra can be found, i.e.\ preliminary group classification can be carried out. The framework of invariant parameterization involving methods of direct group classification is depicted in Figure~\ref{fig:scheme}.

\begin{figure}[!htb]
\centering
\includegraphics[scale=0.8]{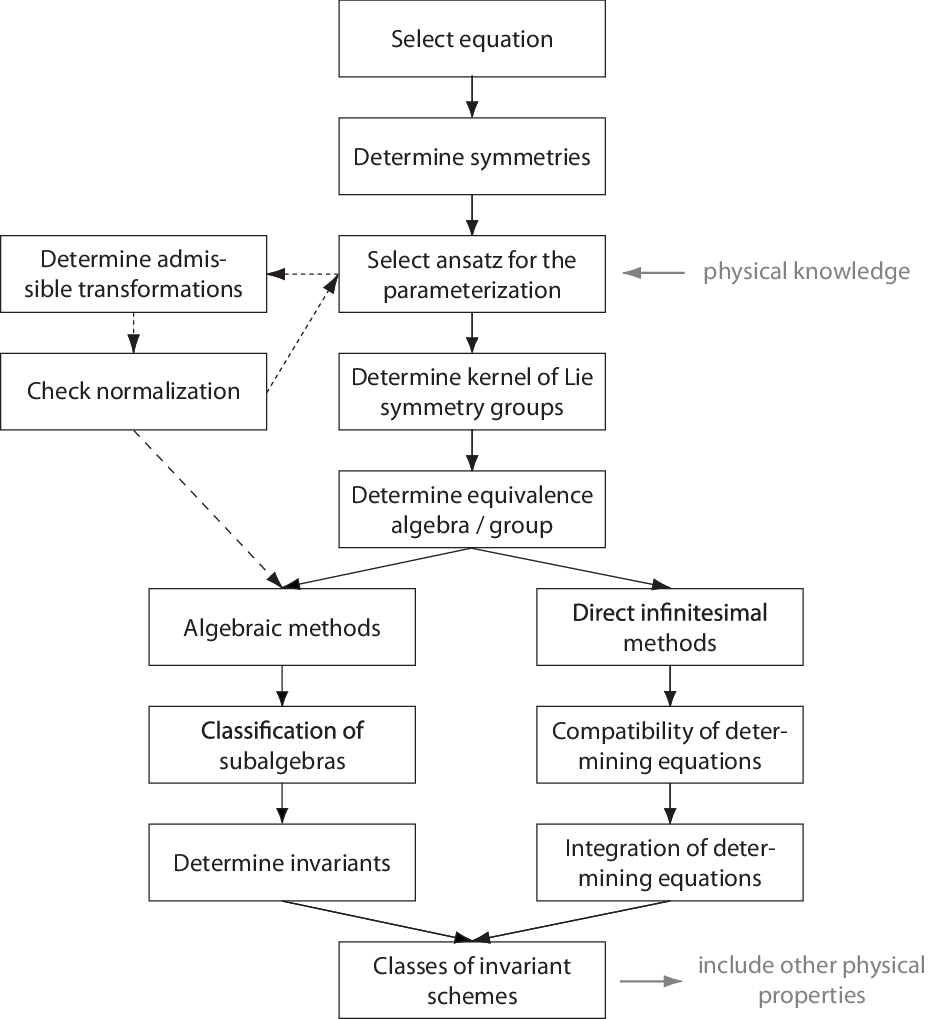}
\caption{Schematic overview of the construction of invariant parameterization schemes based on methods of direct group classification.
\label{fig:scheme}}
\end{figure}

Irrespectively of whether one uses direct or inverse group classification techniques, the procedure of invariant parameterization in fact yield classes of parameterization ansatzes rather than a particular fixed parameterization. This gives a certain degree of freedom which allows one to include other desirable physical or structural features into the parameterization scheme. For example, the specification of the parameterizations in Tables~\ref{tab:OneDimensionalSymmetryExtensionsRestrictedCase}--\ref{tab:OneDimensionalSymmetryExtensionsGeneralCase} can be done by prescribing a particular form of the functions $I_1$ and $I_2$. In the case of inverse group classification, one has to formulate a precise functional relation among related differential invariants. From the point of view of application the freedom in tuning a parameterization is extremely important as the preservation of symmetries is only one feature that might be required when parameterizing a given subgrid-scale process.

Since the primary aim of this paper is a clear presentation of the variety of invariant parameterization methods, we focused on rather simple first order local closure schemes for the classical barotropic vorticity equation, cf.\ the introduction of Section~\ref{sec:vorticity}. That is, we parameterized already the eddy vorticity flux $\overline{\vv'\zeta'}$ using $\bar\zeta$ and its derivatives. Admittedly, this is a quite simple ansatz for one of the simplest physically relevant models in geophysical fluid dynamics. On the other hand, it can be seen that already for this particular simple example the computations involved were rather elaborate. This is in particular true for the computation of the set of admissible transformations for the various classes of vorticity equations considered in Section~\ref{sec:AdmTrans}. Needless to say that irrespectively of practical computational problems the same technique would be applicable to higher order closure schemes as well. In designing such schemes it is necessary to explicitly include differential equations for the first or higher order correlation terms. In the case of the vorticity equations, a second order closure schemes is obtainable upon retaining the equations governing the evolution of $\overline{\vv' \zeta'}$ and parameterize the higher order correlation terms arising in these equations. In practice, however, it becomes increasingly difficult to acquire real atmospheric data for such higher-order correlation quantities, which therefore makes it difficult to propose parameterization schemes based solely on physical considerations~\cite{stul88Ay}. We argue that especially in such cases symmetries could provide a useful guiding principle to determine general classes of relevant parameterizations.

Up to now, we have restricted ourselves to the problem of invariant local closure schemes. Nonlocal schemes constructed using symmetry arguments should be investigated in a subsequent work. This extension to nonlocal parameterization schemes is crucial in order to make general methods available that can be used in the development of parameterization schemes for other types of physical processes in atmosphere-ocean dynamics, including e.g.\ convection. A further perspective for generalization of the present work is the design of parameterization schemes that preserve conservation laws. This is another aspect that is of major importance in practical applications. For parameterizations of conservative processes, it is crucial that the corresponding closed differential equation preserves energy conservation. This is by no means self-evident. In fact, energy conservation is violated by various classes of down-gradient ansatzes~\cite{vall88Ay}, which is straightforward to check also for parameterizations constructed in this paper. The construction of parameterization schemes that retain conservation laws will call for the classification of conservation laws in the way similar as the usual group classification. A main complication is that there is no restriction on the order of conservation laws for general systems of partial differential equations (so far, such restrictions are only known for $(1+1)$-dimensional evolution equations of even order and some similar classes of equations).
The combination of invariant and conservative parameterization schemes is also conceivable. As shown in~\cite{bihl11Fy}, it works for the barotropic vorticity equation on the beta-plane.

It is beyond the scope of the present paper to explicitly test the various parameterization schemes proposed though it was indicated above that some of them might have a physical importance whereas other schemes are obviously flawed. An example on the application of invariant parameterization schemes for the barotropic vorticity equation on the beta-plane to the problem of two-dimensional freely-decaying turbulence has been presented in~\cite{bihl11Fy}.

\section*{Acknowledgments}

The authors are grateful to Sergey Golovin and Michael Kunzinger for useful discussions and interesting comments.
The remarks of the anonymous referee are appreciated.
The research of ROP was supported by the Austrian Science Fund (FWF), project P20632.
AB was a recipient of a DOC-fellowship of the Austrian Academy of Sciences. AB also acknowledges support through the Austrian Science Fund, project J3182--N13.
The present work is carried out within the framework of COST Action ES0905.

\appendix

\section{\protect Inequivalent one-dimensional subalgebras of the equivalence algebra of class~\eqref{eq:vortdirect}}\label{sec:appendix1}

In this appendix, we classify one-dimensional subalgebras of the equivalence algebra~$\mathfrak g^\sim_1$ with basis elements~\eqref{eq:equivalgebravort}. For this means, we subsequently present the commutator table of~$\mathfrak g^\sim_1$. In what follows we omit tildes in the notation of operators.

\begin{table}[!htb]
\footnotesize
\begin{center}
\caption{Commutation relations for the algebra~$\mathfrak g^\sim_1$\label{tab:CommutationRelations}}
\begin{tabular}{|c|ccccc|}
\hline
 & $\DDD_1$ & $\DDD_2$ & $\p_t$ & $\JJ(\beta)$ & $\XX(\gamma^1)$   \\
\hline
$\DDD_1$ & $0$ & $0$ & ${ }$\ \ \,\quad$-\p_t$\quad\ \ \,${ }$ & $\JJ(t\beta_t)$ & $\XX(t\gamma^1_t) $
\\
$\DDD_2$ & $0$ & $0$ & $0$ & $0$ & $-\XX(\gamma^1)$
\\
$\p_t$ & $\p_t$ & $0$ & $0$ & $\JJ(\beta_t)$ &  $\XX(\gamma^1_t)$
\\
$\JJ(\tilde\beta)$ & $-\JJ(t\tilde\beta_t)$ & $0$ & $-\JJ(\tilde\beta_t)$ & $0$ & $-\YY(\tilde\beta\gamma^1)+\mathcal G(\gamma^1\tilde\beta_{tt}y)$
\\
$\XX(\tilde\gamma^1)$ & $-\XX(t\tilde\gamma^1_t) $  & $\XX(\tilde\gamma^1)$ & $-\XX(\tilde\gamma^1_t)$ & $\YY(\beta\tilde\gamma^1)-\mathcal G(\tilde\gamma^1\beta_{tt}y)$ & $0$
\\
$\YY(\tilde\gamma^2)$ & $-\YY(t\tilde\gamma^2_t)$ & $\YY(\tilde\gamma^2)$ & $-\YY(\tilde\gamma^2_t)$ & $\!-\XX(\beta\tilde\gamma^2)+\mathcal G(\tilde \gamma^2\beta_{tt}x)\!$ & $-\ZZ((\gamma^1\tilde\gamma^2)_t)$
\\
$\mathcal R(\tilde\sigma)$ & $-\mathcal R(t\tilde\sigma_t+\tilde\sigma)$ & $0$ & $-\mathcal R(\tilde\sigma_t)$ & $0$ & $-\mathcal H(\gamma^1\tilde\sigma x)+\mathcal G(\gamma^1\tilde\sigma_ty)$
\\
$\mathcal H(\tilde\delta)$ & $-\mathcal H(t\tilde\delta_t+\tilde\delta)$ & $-\mathcal H(x\tilde\delta_x+y\tilde\delta_y-2\tilde\delta)$ & $-\mathcal H(\tilde\delta_t)$ & $-\mathcal H(\beta x\tilde\delta_y-\beta y\tilde\delta_x)$&  $-\mathcal H(\gamma^1\tilde\delta_x)$
\\
$\mathcal G(\tilde\rho)$ & $-\mathcal G(t\tilde\rho_t+2\tilde\rho)$ & $-\mathcal G(x\tilde\rho_x + y\tilde\rho_y+\tilde\rho)$ & $-\mathcal G(\tilde \rho_t)$ & $-\mathcal G(\beta x\tilde\rho_y-\beta y\tilde\rho_x)$ & $\mathcal G(\gamma^1\tilde\rho_x)$
\\
$\ZZ(\tilde\chi)$ & $-\ZZ(t\tilde\chi_t+\tilde\chi)$ & $2\ZZ(\tilde\chi)$ & $-\ZZ(\tilde\chi_t)$ & $0$ & $0$ \\
\hline
\end{tabular}
\end{center}
\begin{center}
\begin{tabular}{|c|ccccc|}
\hline
& $\YY(\gamma^2)$ & $\mathcal R(\sigma)$ & $\mathcal H(\delta)$ & $\mathcal G(\rho)$ & $\ZZ(\chi)$\\
\hline
$\DDD_1$ & $\YY(t\gamma^2_t)$ & $\mathcal R(t\sigma_t+\sigma) $ & $\mathcal H(t\delta_t+\delta)$ & $\mathcal G(t\rho_t+2\rho)$ & $\ZZ(t\chi_t+\chi)$
\\
$\DDD_2$ & $-\YY(\gamma^2)$  & $0$ &$H(x\delta_x+y\delta_y-2\delta)$ & $\mathcal G(x\rho_x+y\rho_y+\rho)$ & $-2\ZZ(\chi)$
\\
$\p_t$ & $\YY(\gamma^2_t)$ & $\mathcal R(\sigma_t)$ & $\mathcal H(\delta_t)$ & $\mathcal G(\rho_t)$ & $\ZZ(\chi_t)$
\\
$\JJ(\tilde\beta)$ & $\XX(\tilde\beta\gamma^2)-\mathcal G(\gamma^2\tilde\beta_{tt}x) $ & $0$ & $\mathcal H(\tilde\beta x\delta_y-\tilde\beta y\delta_x)$ & $\mathcal G(\tilde\beta x\rho_y-\tilde\beta y\rho_x)$ & $0$
\\
$\XX(\tilde\gamma^1)$ & $\ZZ((\tilde\gamma^1\gamma^2)_t)$ & $\mathcal H(\tilde\gamma^1\sigma x)-\mathcal G(\tilde\gamma^1\sigma_ty)$ & $\mathcal H(\tilde \gamma^1\delta_x)$ & $\mathcal G(\tilde\gamma^1\rho_x)$ & $0$
\\
$\YY(\tilde\gamma^2)$ & $0$ & $\mathcal H(\tilde\gamma^2\sigma y)+\mathcal G(\tilde\gamma^2\sigma_t x) $ & $\mathcal H(\tilde \gamma^2\delta_y)$ & $\mathcal G(\tilde\gamma^2\rho_y)$ & $0$
\\
$\mathcal R(\tilde\sigma)$ & $\!-\mathcal H(\gamma^2\tilde\sigma y)-\mathcal G(\gamma^2\tilde\sigma_t x)\!$ & $0$ & $0$ & $0$ & $0$
\\
$\mathcal H(\tilde\delta)$ & $-\mathcal H(\gamma^2\tilde \delta_y)$ & $0$ & $0$ & $0$ & $0$
\\
$\mathcal G(\tilde\rho)$ & $-\mathcal G(\gamma^2\tilde\rho_y)$ & $0$ & $0$ & $0$ & $0$
\\
$\ZZ(\tilde\chi)$ & $0$ & $0$ & $0$ & $0$ & $0$
\\
\hline
\end{tabular}
\end{center}
\end{table}

Based on Table~\ref{tab:CommutationRelations}, it is straightforward to recover the following nontrivial adjoint actions:
\begin{align*}
 &\Ad(e^{\ve\p_t})\DDD_1 = \DDD_1 -\ve\p_t, &                                                              &\Ad(e^{\ve\DDD_1})\p_t = e^{\ve}\DDD_1, \\
 &\Ad(e^{\ve\JJ(\beta)})\DDD_1 = \DDD_1 + \ve\JJ(t\beta_t), &                                              &\Ad(e^{\ve\DDD_1})\JJ(\beta) = \JJ(\beta(e^{-\ve}t)),\\
 &\Ad(e^{\ve\XX(\gamma^1)})\DDD_1 = \DDD_1 + \ve\XX(t\gamma^1_t), &                                        &\Ad(e^{\ve\DDD_1})\XX(\gamma^1) = \XX(\gamma^1(e^{-\ve}t)), \\
 &\Ad(e^{\ve\YY(\gamma^2)})\DDD_1 = \DDD_1 +\ve\YY(t\gamma^2_t), &                                         &\Ad(e^{\ve\DDD_1})\YY(\gamma^2) = \YY(\gamma^2(e^{-\ve}t)), \\
 &\Ad(e^{\ve\mathcal R(\sigma)})\DDD_1 = \DDD_1 + \ve\mathcal R(t\sigma_t+\sigma), &                       &\Ad(e^{\ve\DDD_1})\mathcal R(\sigma) = \mathcal R(e^{-\ve}\sigma(e^{-\ve}t)), \\
 &\Ad(e^{\ve\mathcal H(\delta)})\DDD_1 = \DDD_1 + \ve\mathcal H(t\delta_t+\delta), &                       &\Ad(e^{\ve\DDD_1})\mathcal H(\delta) = \mathcal H(e^{-\ve}\delta(e^{-\ve}t,x,y)), \\
 &\Ad(e^{\ve\mathcal G(\rho)})\DDD_1 = \DDD_1 + \ve\mathcal G(t\rho_t+2\rho), &                            &\Ad(e^{\ve\DDD_1})\mathcal G(\rho) = \mathcal G(e^{-2\ve}\rho(e^{-\ve}t,x,y)), \\
 &\Ad(e^{\ve\ZZ(\chi)})\DDD_1 = \DDD_1 + \ve\ZZ(t\chi_t+\chi), &                                           &\Ad(e^{\ve\DDD_1})\ZZ(\chi) = \ZZ(e^{-\ve}\chi(e^{-\ve}t)), \\
 &\Ad(e^{\ve\JJ(\beta)})\p_t = \p_t + \ve\JJ(\beta_t), &                                                   &\Ad(e^{\ve\p_t})\JJ(\beta) = \JJ(\beta(t-\ve)), \\
 &\Ad(e^{\ve\XX(\gamma^1)})\p_t = \p_t + \ve\XX(\gamma^1_t), &                                             &\Ad(e^{\ve\p_t})\XX(\gamma^1) = \XX(\gamma^1(t-\ve)),\\
 &\Ad(e^{\ve\YY(\gamma^2)})\p_t = \p_t + \ve\YY(\gamma^2_t), &                                             &\Ad(e^{\ve\p_t})\YY(\gamma^2) = \YY(\gamma^2(t-\ve)), \\
 &\Ad(e^{\ve\mathcal R(\sigma)})\p_t = \p_t + \ve\mathcal R(\sigma_t),&                                    &\Ad(e^{\ve\p_t})\mathcal R(\sigma) = \mathcal R(\sigma(t-\ve)), \\
 &\Ad(e^{\ve\mathcal H(\delta)})\p_t = \p_t + \ve\mathcal H(\delta_t), &                                   &\Ad(e^{\ve\p_t})\mathcal H(\delta) = \mathcal H(\delta(t-\ve,x,y)), \\
 &\Ad(e^{\ve\mathcal G(\rho)})\p_t = \p_t + \ve\mathcal G(\rho_t),&                                        &\Ad(e^{\ve\p_t})\mathcal G(\rho) = \mathcal G(\rho(t-\ve,x,y)), \\
 &\Ad(e^{\ve\ZZ(\chi)})\p_t = \p_t + \ve\ZZ(\chi_t), &                                                     &\Ad(e^{\ve\p_t})\ZZ(\chi) = \ZZ(\chi(t-\ve)), \\
 &\Ad(e^{\ve\XX(\gamma^1)})\DDD_2 = \DDD_2 - \ve\XX(\gamma^1), &                                           &\Ad(e^{\ve\DDD_2})\XX(\gamma^1) = \XX(e^\ve\gamma^1), \\
 &\Ad(e^{\ve\YY(\gamma^2)})\DDD_2 = \DDD_2 - \ve\YY(\gamma^2), &                                           &\Ad(e^{\ve\DDD_2})\YY(\gamma^2) = \YY(e^\ve\gamma^2), \\
 &\Ad(e^{\ve\mathcal H(\delta)})\DDD_2 = \DDD_2 + \ve\mathcal H(x\delta_x+y\delta_y-2\delta), &            &\Ad(e^{\ve\DDD_2})\mathcal H(\delta) = \mathcal H(e^{2\ve}\delta(t,e^{-\ve}x,e^{-\ve}y)), \\
 &\Ad(e^{\ve\mathcal G(\rho)})\DDD_2 = \DDD_2 + \ve\mathcal G(x\rho_x+y\rho_y+\rho), &                     &\Ad(e^{\ve\DDD_2})\mathcal G(\rho) = \mathcal G(e^{-\ve}\rho(t,e^{-\ve}x,e^{-\ve}y)) , \\
 &\Ad(e^{\ve\ZZ(\chi)})\DDD_2 = \DDD_2 - 2\ve\ZZ(\chi), &                                                  &\Ad(e^{\ve\DDD_2})\ZZ(\chi) = \ZZ(e^{2\ve}\chi), \\
 &\Ad(e^{\ve\XX(\gamma^1)})\JJ(\beta) = A_1, &                                                             &\Ad(e^{\ve\JJ(\beta)})\XX(\gamma^1) = A_3,  \\
 &\Ad(e^{\ve\YY(\gamma^2)})\JJ(\beta) = A_2, &                                                             &\Ad(e^{\ve\JJ(\beta)})\YY(\gamma^2) = A_4, \\
 &\Ad(e^{\ve\mathcal H(\delta)})\JJ(\beta) = \JJ(\beta) + \ve\mathcal H(\beta x\delta_y-\beta y\delta_x), &&\Ad(e^{\ve\JJ(\beta)})\mathcal H(\delta) = A_5, \\
 &\Ad(e^{\ve\mathcal G(\rho)})\JJ(\beta) = \JJ(\beta) + \ve\mathcal G(\beta x\rho_y-\beta y\rho_x), &      &\Ad(e^{\ve\JJ(\beta)})\mathcal G(\rho) = A_6, \\
 &\Ad(e^{\ve\YY(\gamma^2)})\XX(\gamma^1)=\XX(\gamma^1) +\ve\ZZ((\gamma^1\gamma^2)_t), &                    &\Ad(e^{\ve\XX(\gamma^1)})\YY(\gamma^2) = \YY(\gamma^2) - \ve\ZZ((\gamma^1\gamma^2)_t), \\
 &\Ad(e^{\ve\mathcal R(\sigma)})\XX(\gamma^1) = A_7, &                                                     &\Ad(e^{\ve\XX(\gamma^1)})\mathcal R(\sigma) = A_8, \\
 &\Ad(e^{\ve\mathcal H(\delta)})\XX(\gamma^1) = \XX(\gamma^1) + \ve\mathcal H(\gamma^1\delta_x),&          &\Ad(e^{\ve\XX(\gamma^1)})\mathcal H(\delta) = \mathcal H(\delta(t,x-\ve\gamma^1,y)), \\
 &\Ad(e^{\ve\mathcal G(\rho)})\XX(\gamma^1) = \XX(\gamma^1) + \ve\mathcal G(\gamma^1\rho_x), &             &\Ad(e^{\ve\XX(\gamma^1)})\mathcal G(\rho) = \mathcal G(\rho(t,x-\ve\gamma^1,y)), \\
 &\Ad(e^{\ve\mathcal R(\sigma)})\YY(\gamma^2) = A_9, &                                                     &\Ad(e^{\ve\YY(\gamma^2)})\mathcal R(\sigma) = A_{10}, \\
 &\Ad(e^{\ve\mathcal H(\delta)})\YY(\gamma^2) = \YY(\gamma^2) +\ve\mathcal H(\gamma^2\delta_y), &          &\Ad(e^{\ve\YY(\gamma^2)})\mathcal H(\delta) = \mathcal H(\delta(t,x,y-\ve\gamma^2)), \\
 &\Ad(e^{\ve\mathcal G(\rho)})\YY(\gamma^2) = \YY(\gamma^2) +\ve\mathcal G(\gamma^2\rho_y), &              &\Ad(e^{\ve\YY(\gamma^2)})\mathcal G(\rho) = \mathcal G(\rho(t,x,y-\ve\gamma^2)),
\end{align*}
where
\begin{align*}
 &A_1:= \JJ(\beta) - \ve\bigl(\YY(\beta\gamma^1)-\mathcal G(\beta_{tt}\gamma^1 y)\bigr) + \tfrac12\ve^2\ZZ\bigl((\beta(\gamma^1)^2)_t\bigr), \\
 &A_2:= \JJ(\beta) + \ve\bigl(\XX(\beta\gamma^2)-\mathcal G(\beta_{tt}\gamma^2 x)\bigr) + \tfrac12\ve^2\ZZ\bigl((\beta(\gamma^2)^2)_t), \\
 &A_3:= \XX(\gamma^1\cos\beta\ve)+\YY(\gamma^1\sin\beta\ve) - \ve\mathcal G\bigl(\gamma^1\beta_{tt}(-x\sin\beta\ve+y\cos\beta\ve)\bigr), \\
 &A_4:=-\XX(\gamma^2\sin\beta\ve)+\YY(\gamma^2\cos\beta\ve) + \ve\mathcal G\bigl(\gamma^1\beta_{tt}(x\cos\beta\ve+y\sin\beta\ve)\bigr), \\
 &A_5:= \mathcal H(\delta(t,x\cos\beta\ve+y\sin\beta\ve,-x\sin\beta\ve+y\cos\beta\ve)), \\
 &A_6 := \mathcal G(\rho(t,x\cos\beta\ve+y\sin\beta\ve,-x\sin\beta\ve+y\cos\beta\ve)), \\
 &A_7 := \XX(\gamma^1) + \ve\bigl(\mathcal H(\gamma^1\sigma x)-\mathcal G(\gamma^1\sigma_t y)\bigr), \\
 &A_8:= \mathcal R(\sigma) -\ve\left(\mathcal H(\gamma^1\sigma x)-\mathcal G(\gamma^1\sigma_t y)\right)+\tfrac12\ve^2\mathcal H\left((\gamma^1)^2\sigma\right), \\
 &A_9 := \YY(\gamma^2) + \ve\left(\mathcal H\left(\gamma^2\sigma y\right)+ \mathcal G(\gamma^2\sigma_t x)\right), \\
 &A_{10}:= \mathcal R(\sigma) -\ve\left(\mathcal H(\gamma^2\sigma y)+\mathcal G(\gamma^2\sigma_t x)\right)+\tfrac12\ve^2\mathcal H\left((\gamma^2)^2\sigma\right).
\end{align*}
Using the above adjoint actions, we construct the following optimal list of inequivalent one-dimensional subalgebras of~$\mathfrak g^\sim_1$:
\begin{align}\label{eq:OneDimensionalSubalgebrasEquivalenceAlgebra}
\begin{split}
 &\langle\DDD_1+a\DDD_2\rangle,\quad \langle\p_t+b\DDD_2\rangle, \quad \langle\DDD_2+\JJ(\beta)+\mathcal R(\sigma)\rangle, \quad\langle\JJ(\beta)+\mathcal R(\sigma)+\ZZ(\chi)\rangle,  \\
 &\langle\XX(\gamma^1)+\mathcal R(\sigma)\rangle,\quad \langle\mathcal R(\sigma)+\mathcal H(\delta)+\mathcal G(\rho)+\ZZ(\chi) \rangle,
\end{split}
\end{align}
where $a\in\mathbb R$, $b\in\{-1,0,1\}$. In fact, each element of the above list represents a parameterized class of subalgebras rather than a single subalgebra. Particular subalgebras correspond to arbitrary but fixed values of parameters.
Subalgebras within each of the four last classes can be equivalent.
Thus, in the third class we can use adjoint action $\Ad(e^{\ve\DDD_1})$ to rescale $\sigma$ as well as the argument $t$ of~$\beta$ and~$\sigma$. Using $\Ad(e^{\ve\p_t})$ allows us to shift $t$ in the functions $\beta$ and $\sigma$.
In the fourth class, equivalence is understood up to actions of $\Ad(e^{\ve\DDD_1})$, $\Ad(e^{\ve\DDD_2})$ and $\Ad(e^{\ve\p_t})$, which permit rescaling of $\sigma$, $\chi$ and their argument $t$, scaling of $\chi$ as well as shifts of $t$ in $\beta$, $\sigma$ and $\chi$. Similar equivalence is also included in the fifth class. The last class comprises equivalence with respect to actions of $\Ad(e^{\ve\JJ(\beta)})$, $\Ad(e^{\ve\XX(\gamma^1)})$ and $\Ad(e^{\ve\YY(\gamma^2)})$.
In the three last classes we can also rescale the entire basis elements.

{\footnotesize\setlength{\itemsep}{0ex}

}

\end{document}